\documentclass[english,final]{IEEEtran}

\usepackage[T1]{fontenc}
\usepackage{amsthm}
\usepackage{amsmath}
\usepackage{amssymb}
\usepackage{graphicx}
\usepackage{dsfont}
\usepackage{cite}
\usepackage{amsmath}
\usepackage{array}
\usepackage{multirow}
\usepackage{caption}
\usepackage{color}

\usepackage{cite}
\usepackage{amsmath,amssymb,amsthm,dsfont,mathrsfs}  
\usepackage{dsfont}
\usepackage{graphicx} 
\usepackage{hyperref}
\hypersetup{colorlinks=true,linkcolor=blue,citecolor=blue,linktocpage}

\usepackage{amsmath,amssymb,amsthm,dsfont,mathrsfs}
\usepackage{makecell}
\usepackage{array,multirow,graphicx}
\usepackage{setspace}
\newcommand{\E}{\mathbb{E}}
\newcommand{\Var}{\mathrm{Var}}
\newcommand{\Cov}{\mathrm{Cov}}

\newtheorem{theorem}{Theorem}
\newtheorem{corollary}{Corollary}
\newtheorem{lemma}{Lemma}
\theoremstyle{definition}
\newtheorem{definition}{Definition}
\newtheorem{remark}{Remark}

\usepackage{algorithm}
\usepackage[noend]{algorithmic}

\begin{document}

\title{Performance Bounds for Group Testing \\ With Doubly-Regular Designs}

\author{Nelvin Tan, Way Tan, and Jonathan Scarlett\thanks{

N.~Tan is with the Department of Engineering, University of Cambridge.
W.~Tan and J.~Scarlett are with the Department of Computer Science, National University of Singapore (NUS), and also with the Department of Mathematics, NUS. e-mails: \url{tcnt2@cam.ac.uk}; \url{e0174826@u.nus.edu}; \url{scarlett@comp.nus.edu.sg}.

This work was presented in part at the 2021 IEEE International Symposium on Information Theory (ISIT) \cite{Nel21}.

This work was supported by an NUS Early Career Research Award.
N.~Tan and W.~Tan contributed equally to this work.
}}

\allowdisplaybreaks

\maketitle

\begin{abstract}
    In the group testing problem, the goal is to identify a subset of defective items within a larger set of items based on tests whose outcomes indicate whether any defective item is present. This problem is relevant in areas such as medical testing, DNA sequencing, and communications.  In this paper, we study a doubly-regular design in which the number of tests-per-item and the number of items-per-test are fixed.  We analyze the performance of this test design alongside the Definite Defectives (DD) decoding algorithm in several settings, namely, (i) the sub-linear regime $k=o(n)$ with exact recovery, (ii) the linear regime $k=\Theta(n)$ with approximate recovery, and (iii) the size-constrained setting, where the number of items per test is constrained. Under setting (i), we show that our design together with the DD algorithm, matches an existing achievability result for the DD algorithm with the near-constant tests-per-item design, which is known to be asymptotically optimal in broad scaling regimes. Under setting (ii), we provide novel approximate recovery bounds that complement a hardness result regarding exact recovery. Lastly, under setting (iii), we improve on the best known upper and lower bounds in scaling regimes where the maximum test size grows with the total number of items.
\end{abstract}

\newcounter{mytempeqncnt}

\section{Introduction}

In the group testing problem, the goal is to identify a small subset of defective items of size $k$ within a larger set of items of size $n$, based on a number $T$ of tests. This problem is relevant in areas such as medical testing, DNA sequencing, and communication protocols \cite[Sec.~1.7]{Ald19}, and has recently found utility in COVID-19 testing \cite{Ald21}.

In non-adaptive group testing, the placements of items into test can be represented by a binary test matrix of size $T \times n$.  The strongest known theoretical guarantees are based on the idea of generating this matrix at random and analyzing the average performance.  Starting from early studies of group testing, a line of works led to a detailed understanding of the i.i.d.~test design \cite{Mal78,Ati12,Ald14a,Sca15b}, and more recently, further improvements were shown for a near-constant tests-per-item design \cite{Joh16,Coj19}, whose asymptotic optimality in sub-linear sparsity regimes was established in \cite{Coj20}.  Recently, there has been increasing evidence that {\em doubly-regular designs} (i.e., both constant tests-per-item and items-per-test) also play a crucial role in various settings of interest:
\begin{itemize}
    \item The work of Mez\'ard {\em et al.} \cite{Mez08} uses heuristic arguments from statistical physics to suggest that doubly-regular designs achieve the same optimal threshold as the near-constant tests-per-item design when $k = \Theta(n^{\theta})$, at least when $\theta$ is not too small.
    \item In constrained settings where the number of items per test cannot exceed a pre-specified threshold, doubly-regular designs have been used to obtain performance bounds that appear to be difficult or impossible to obtain using the other designs mentioned above \cite{Ven19,Oli20}.
    \item In the linear sparsity regime (i.e., $k = \Theta(n)$), various two-stage adaptive designs were studied in \cite{Ald20}, and using a doubly-regular design in the first stage led to strict improvements over the other designs.  See also Appendix \ref{app:COMP} for analogous observations with non-adaptive testing and approximate recovery.
    \item From a more practical viewpoint, doubly-regular designs have found utility in application-driven settings, e.g., see \cite{Goe21,Gho21} for recent studies relating to medical testing.
\end{itemize}
In this paper, motivated by these developments, we seek to provide a more detailed understanding of non-adaptive doubly-constant test designs, particularly when paired with the Definite Defectives (DD) algorithm \cite{Ald14a}.  Briefly, our contributions are as follows: (i) For $\theta \in \big[\frac{1}{2},1\big)$, we rigorously prove the above-mentioned result shown heuristically in \cite{Mez08}, albeit with a slightly different version of the doubly-regular design; (ii) We establish new performance bounds for non-adaptive group testing in the linear regime ($k = \Theta(n)$) with some false negatives allowed in the reconstruction, complementing strong impossibility results for exact recovery \cite{Ald18,Bay20}; (iii) We provide improved upper and lower bounds on the number of tests for the constrained setting with at most $\rho$ items per test.  Our bounds apply to general scaling regimes beyond the regime $\rho = O(1)$ recently studied in \cite{Oli20}, and our consideration of the DD algorithm leads to strict improvements over the COMP algorithm considered in \cite{Ven19}.


\subsection{Problem Setup}

Let $n$ denote the number of items, which we label as $[n] = \{1,\dots,n\}$. Let $\mathcal{K}\subset [n]$ denote the fixed set of defective items, and let $k=|\mathcal{K}|$ be the number of defective items. We adopt the combinatorial prior \cite{Ald19}, where $\mathcal{K}$ is chosen uniformly at random from ${n\choose k}$ sets of size $k$. We let $T=T(n)$ be the number of tests performed. The $i$-th test takes the form 
\begin{align}
    Y^{(i)}=\bigvee_{j\in\mathcal{K}}X_j^{(i)}, \label{eq:test_outcome_formula}
\end{align}
where the test vector $X^{(i)}=\big(X_1^{(i)},\dots,X_n^{(i)}\big)\in\{0,1\}^n$ indicates which items are are included in the test, and $Y^{(i)}\in\{0,1\}$ is the resulting observation, which indicates whether at least one defective item was included in the test. The goal of group testing is to design a sequence of tests $X^{(1)},\dots,X^{(T)}$, with $T$ ideally as small as possible, such that the outcomes can be used to reliably recover the defective set $\mathcal{K}$. We focus on the non-adaptive setting, in which all tests $X^{(1)},\dots,X^{(T)}$ must be designed prior to observing any outcomes. 

Next, we introduce the main defining features distinguishing the settings we consider:
\begin{itemize}
    \item Regarding the scaling of $k$, we consider the following:
    \begin{itemize}
        \item \textbf{Sub-linear:} We have $k=\Theta\big(n^\theta\big)$ for some constant $\theta\in(0,1)$. This is the regime where defectivity is rare, and also where group testing typically exhibits the greatest gains.
        \item \textbf{Linear:} We have $k=pn$ for some prevalence rate $p\in(0,1)$. This regime may potentially be of greater relevance in certain practical situations, e.g., with $p$ representing the prevalence of a disease.
    \end{itemize}
    \item Regarding the constraints (or lack thereof), we consider the following:
    \begin{itemize}
        \item \textbf{Unconstrained:} There is no restriction on the number of tests-per-item or the number of items-per-test in the test design.
        \item \textbf{Size-constrained:} Tests are size-constrained and thus contain no more than $\rho=\Theta\big(\big(\frac{n}{k}\big)^\beta\big)$ items per test, for some constant $\beta\in(0,1)$. Note that if each test comprises of $\Theta(n/k)$ items, then $\Theta(k\log n)$ tests suffice for group testing algorithms with asymptotically vanishing error probability \cite{Cha14,Ald14a,Sca15b,Joh16}. One can alternatively  use exactly $n$ tests via one-by-one testing, and it has recently been shown that taking the better of the two (i.e., $\Theta(\min\{k \log n, n\})$ tests) gives asymptotically optimal scaling \cite{Bay20}. Hence, to avoid essentially reducing to the unconstrained setting, the parameter regime of interest in the size-constrained setting is $\rho=o(n/k)$, which justifies the scaling $\rho=\Theta\big(\big(\frac{n}{k}\big)^\beta\big)$.
    \end{itemize}
    \item Regarding recovery criteria, we consider the following:
    \begin{itemize}
        \item \textbf{Exact recovery:} We seek to develop a testing strategy and decoder that produces an estimate $\widehat{\mathcal{K}}$ such that the error probability $P_e=\mathbb{P}\big[\widehat{\mathcal{K}}\neq\mathcal{K}\big]$ is asymptotically vanishing as $n \to \infty$.
        \item \textbf{Approximate recovery:} We seek to characterize the per-item false-positive rate (FPR) and false-negative rate (FNR), which are defined as the probability that a non-defective item (picked uniformly at random) is declared defective (i.e., ${\rm FPR} = \frac{\E[|\widehat{\mathcal{K}}\setminus\mathcal{K}|]}{|[n]\setminus\mathcal{K}|}$), and the probability that a defective item (picked uniformly at random) is declared non-defective (i.e., ${\rm FNR} = \frac{\E[|\mathcal{K}\setminus\widehat{\mathcal{K}}|]}{|\mathcal{K}|}$), respectively.
    \end{itemize}
\end{itemize}

With these definitions in place, the settings that we focus on are (i) the unconstrained sub-linear regime with exact recovery, (ii) the unconstrained linear regime with approximate recovery, and (iii) the size-constrained sub-linear regime with exact recovery. More specifically, in the second of these, we only consider the FNR; the FPR is not considered, since the DD algorithm that we study never declares a non-defective item to be defective.\footnote{See also Appendix \ref{app:COMP} for a result regarding the COMP algorithm with only false positives and no false negatives.}  While the above definitions lead to $2^3 = 8$ possible settings of interest, our focus is on three that we believe to be suitably representative and of the most interest given what is already known in existing works (e.g., due to the hardness of exact recovery in the linear regime \cite{Ald18,Bay20}).

\textbf{Notation.} Throughout the paper, the function $\log(\cdot)$ has base $e$, and we make use of Bachmann-Landau asymptotic
notation (i.e., $O$, $o$, $\Omega$, $\omega$, $\Theta$).

\subsection{Related Work}

\begin{algorithm}[t]
    \begin{algorithmic}[1]
        \REQUIRE $T$ tests.
        \STATE Initialize two empty sets $\mathcal{PD}$ (possibly defective set) and $\mathcal{DD}$ (definitely defective set).
        \STATE Label any item in a negative test as definitely non-defective, and add all remaining items to $\mathcal{PD}$.
        \FOR{each test} {
        \STATE If the test contains exactly one item from $\mathcal{PD}$, then add that item to the set $\mathcal{DD}$.
        }\ENDFOR
        \RETURN $\widehat{\mathcal{K}} = \mathcal{PD}$ for COMP, or $\widehat{\mathcal{K}} = \mathcal{DD}$ for DD.
    \end{algorithmic}
    \caption{COMP and DD algorithms \cite{Cha14,Ald14a}. \label{alg:COMP&DD_algo}}
\end{algorithm}

We focus on non-adaptive and noiseless group testing with a combinatorial prior. We begin by introducing two common decoding algorithms, Combinatorial Orthogonal Matching Pursuit (COMP) and Definite Defectives (DD), in Algorithm \ref{alg:COMP&DD_algo}. A key difference between the COMP and DD algorithm is that the COMP algorithm produces only false positives (i.e., no false negatives), while the DD algorithm produces only false negatives (i.e., no false positives).

Next, we introduce some test designs \cite[Section 1.3]{Ald19}, which will be useful for purposes of comparison later:
\begin{itemize}
    \item \textbf{Bernoulli design:} Each item is randomly included in each test independently with some fixed probability.
    \item \textbf{Near-constant tests-per-item:} Each item is included in some fixed number of tests, with the tests for each item chosen uniformly at random \textit{with replacement}, independent from the choices for all other items.
    \item \textbf{Constant tests-per-item:} Each item is included in some fixed number of tests, with the tests for each item chosen uniformly at random \textit{without replacement}, independent from the choices for all other items.
    \item \textbf{Doubly-regular design:} Both the number of tests-per-item and the number of items-per-test are fixed to pre-specified values.  Previous works predominantly considered the uniform distribution over all designs satisfying these conditions, though our own results will use a slightly different block-structured variant from \cite{Bro20}.
\end{itemize}

We proceed to review the related work for each setting.

\subsubsection{Unconstrained Sub-Linear Regime} 

In the unconstrained setting with sub-linear sparsity, the following number of tests multiplied with $(1+\epsilon)$ (where $\epsilon$ is any positive constant) are sufficient to attain asymptotically vanishing error probability: 
\begin{itemize}
    \item Bernoulli testing \& COMP decoding \cite{Ald15, cha11}: $ek\log n \approx 2.72k\log n$;
    \item Bernoulli testing \& DD decoding \cite{Ald15, cha11}: $e\max\{\theta,1-\theta\}k\log n \approx 2.72\max\{\theta,1-\theta\}k\log n$;
    \item Near-constant tests-per-item \& COMP decoding \cite{Joh16}: $\frac{k\log n}{\log^22} \approx 2.08k\log n$;
    \item Near-constant tests-per-item \& DD decoding \cite{Joh16}: $\frac{\max\{\theta,1-\theta\}}{\log^22}k\log n \approx 2.08\max\{\theta,1-\theta\}k\log n$.
\end{itemize}
These results indicate that the near-constant tests-per-item is superior to Bernoulli testing, with the intuition being that the former avoids over-testing or under-testing items.  We also observe that DD decoding is superior to COMP decoding, with the intuition being that the information from positive tests is ``wasted'' in the latter.  Further improvements for information-theoretically optimal decoding are discussed below.

Additionally, converse results have been proven for each of these test designs: In the sub-linear regime with unconstrained tests, \textit{any} decoding algorithm that uses the following number of tests multiplied by $1-\epsilon$ (where $\epsilon$ is arbitrarily small) is unable to attain asymptotically vanishing error probability:
\begin{itemize}
    \item Bernoulli testing \cite{Ald15}: $\big(\log2\cdot\max_{\nu>0}\min\big\{H_2(e^{-\nu}), \frac{\nu e^{-\nu}}{\log{2}}\frac{1-\theta}{\theta}\big\}\big)^{-1}(1-\theta)k\log n$;
    \item Near-constant tests-per-item \cite{Joh16,Coj19}: $\max\big\{\frac{\theta}{\log^2 2},\frac{1-\theta}{\log 2}\big\}k\log n$.
\end{itemize}
For both designs, the DD algorithm's performance matches the converse for $\theta \ge \frac{1}{2}$. 

For information-theoretically optimal decoding, exact thresholds on the required number of tests were characterized for Bernoulli testing in \cite{Sca15b,Ald15}, and for the near-constant tests-per-item design in \cite{Joh16,Coj19}.  Perhaps most importantly among these, it was shown in \cite{Coj20} that the above converse for the near-constant tests-per-item design extends to {\em arbitrary} non-adaptive designs, and that a matching achievability threshold holds for a certain spatially-coupled test design with {\em polynomial-time decoding}.  Hence, $\max\big\{\frac{\theta}{\log^2 2},\frac{1-\theta}{\log 2}\big\}k\log n$ is the optimal threshold for all $\theta \in (0,1)$, and for $\theta \ge \frac{1}{2}$ the DD algorithm with near-constant tests-per-item is asymptotically optimal.

M{\'e}zard {\em et al.}~\cite{Mez08} considered doubly-regular designs, and designs with constant tests-per-item only. Their analysis used heuristics from statistical physics to suggest that such designs can improve on Bernoulli designs, and match the above near-constant tests-per-item bound for the DD algorithm.  The analysis in \cite{Mez08} contains some non-rigorous steps; in particular, they make use of a ``no short loops'' assumption that is only verified for $\theta>\frac{5}{6}$ and conjectured for $\theta\geq\frac{2}{3}$, while experimentally being shown to fail for certain smaller values.  One of our contributions in this paper is to establish a rigorous version of their result for $\theta \ge \frac{1}{2}$.

A distinct line of works has sought designs that not only require a low number of tests, but also near-optimal decoding complexity (e.g., $k\,\text{poly}(\log n)$) \cite{Cai13,Lee16,Ngo11,Ind10,Eri20,Nel21a}. However, our focus in this paper is on the required number of tests, for which the existing guarantees of such algorithms contain loose constants or extra logarithmic factors.


\subsubsection{Linear Regime} 

Under the exact recovery guarantee, the known methods for deriving achievability bounds on $T$ in the unconstrained sub-linear regime do not readily extend to the linear regime. In fact, as the following results assert, individual testing is optimal for exact recovery in the linear regime. 
\begin{itemize}
    \item \textbf{Weak converse \cite{Ald18}:} In the linear regime with prevalence $p \in (0,1)$, if we use $T<n-1$ tests, there exists $\epsilon = \epsilon(p) > 0$ independent of $n$ such that $P_e \ge \epsilon$.
    
    \item \textbf{Strong converse \cite{Bay20}:} In the linear regime with any fixed $p \in (0,1)$, if $T \le (1-\epsilon)n$ for some constant $\epsilon > 0$, then $P_e \to 1$ as $n \to \infty$.
\end{itemize}
These results imply that individual testing is an asymptotically optimal non-adaptive strategy for exact recovery. Hence, for this regime, we will instead investigate the FNR.

A recent study of two-stage adaptive algorithms in \cite{Ald20} turns out to be relevant to our setup (albeit not directly applicable in our non-adaptive setup).  There, the high-level approach was the following, which was also considered in earlier works (e.g., see \cite{Mez11,Dya14}): (i) Conduct non-adaptive testing and identify a set of definitely non-defective items, leaving only the possibly defective items. (ii) Conduct individual testing on the remaining possibly defective items.  It was shown in \cite{Ald20} that using a doubly-regular design from \cite{Bro20} in the first stage gives us the lowest expected number of tests required to attain zero error probability, with strict improvements over the near-constant tests-per-item design. This motivates us to analyze the FNR of the DD algorithm with a doubly-regular design.

\subsubsection{Size-Constrained Sub-Linear Regime} \label{sec:prev_results_con_sublinear}


The results most relevant to this setting are described as follows, where $k=\Theta(n^\theta)$ with $\theta\in[0,1)$ throughout:
\begin{itemize}
    \item \textbf{Converse \cite{Ven19}:} For $\rho=\Theta\big(\big(\frac{n}{k}\big)^\beta\big)$ with $\beta\in[0,1)$, and an arbitrarily small $\epsilon>0$, any non-adaptive algorithm with error probability at most $\epsilon$ requires $T\geq\frac{1-6\epsilon}{1-\beta}\cdot\frac{n}{\rho}$, for sufficiently large $n$.
    
    \item \textbf{Improved Converse for $\beta = 0$ \cite{Oli20}:} For $\rho = \Theta(1)$ satisfying $\rho \ge 1+\big\lfloor\frac{\theta}{1-\theta}\rfloor$, if $T\leq(1-\epsilon)\max\big\{\big(1+\big\lfloor\frac{\theta}{1-\theta}\big\rfloor\big)\frac{n}{\rho},\frac{2n}{\rho+1}\big\}$ for some constant $\epsilon > 0$, then any non-adaptive algorithm fails (with $1-o(1)$ probability if $\frac{\theta}{1-\theta}$ is a non-integer, and with $\Omega(1)$ probability otherwise).
    
    \item \textbf{Achievability \cite{Ven19}:} Under a doubly-regular random test design and the COMP algorithm, for $\rho=\Theta\big(\big(\frac{n}{k}\big)^\beta\big)$ with $\beta\in[0,1)$, and an arbitrarily small $\epsilon>0$, the error probability is asymptotically vanishing when $T\geq\big\lceil\frac{1+\epsilon}{(1-\theta)(1-\beta)}\big\rceil\cdot\big\lceil\frac{n}{\rho}\big\rceil$.
    
    \item \textbf{Improved Achievability for $\beta = 0$ \cite{Oli20}:} In the regime $\rho = \Theta(1)$, under a suitably-chosen near-regular random test design (which slightly differs depending on whether or not $\theta \ge \frac{1}{2}$), the error probability is asymptotically vanishing when $T\geq(1+\epsilon)\max\big\{\big(1+\big\lfloor\frac{\theta}{1-\theta}\big\rfloor\big)\frac{n}{\rho},\frac{2n}{\rho+1}\big\}$.  This is achieved using the Sequential COMP (SCOMP) algorithm \cite{Ald14a}, which starts with the DD solution and then iteratively refines it.
\end{itemize}

The above results for $\beta = 0$ (i.e., $\rho = \Theta(1)$) strictly improve on the results for general $\beta$, and enjoy matching achievability and converse thresholds.  Thus, it is natural to ask whether we can attain similar improvements for the case that $\rho=\Theta\big( (n/k)^\beta \big)$, for $\beta=(0,1)$ (i.e., large $\rho$).  We partially answer this question in the affirmative.

Regarding our use of a doubly-regular design, the column weight restriction is not strictly imposed by the testing constraints, but helps in avoiding ``bad'' events where some items are not tested enough (or even not tested at all). For example, in the case of the COMP algorithm, the doubly-regular design helps to reduce the number of tests by a factor of $O(\log n)$ compared to i.i.d.~testing.

Additional results are given for the adaptive setting  in \cite{Nel20, Oli20}, and for the noisy non-adaptive setting in \cite[Section 7.2]{Ven19}. Another notable type of sparsity constraint is that of finitely divisible items (i.e., bounded tests-per-item) constraint, which are studied in \cite{Ven19,Nel20,Oli20}.  The main reason that we do not consider such constraints here is that it was already studied extensively in \cite{Oli20} without any analog of the above-mentioned restrictive assumption $\rho = O(1)$.

\subsection{Contributions}

Our main contributions are as follows:
\begin{itemize}
    \item \textbf{Unconstrained sub-linear regime with exact recovery:} We provide an achievability result for the DD algorithm with a doubly-regular design that matches a result of \cite{Joh16} for the DD algorithm with a near-constant tests-per-item design, which is asymptotically optimal when $\theta\geq \frac{1}{2}$.  Thus, for this range of $\theta$, we provide a rigorous counterpart to the result shown heuristically in \cite{Mez08}.
    \item \textbf{Unconstrained linear regime with approximate recovery:} We provide an asymptotic bound on the FNR for the DD algorithm with a doubly-regular design, and further characterize the low-sparsity limit analytically, while evaluating various higher-sparsity regimes numerically.
    \item \textbf{Size-constrained sub-linear regime with exact recovery:} Motivated by recently-shown gains for the regime $\rho = O(1)$ \cite{Oli20}, we show that analogous gains are also possible for more general $\rho = o\big(\frac{n}{k}\big)$.  We improve on the best known achievability and converse bounds in such regimes, in particular using the DD algorithm to improve over known results for the COMP algorithm.
\end{itemize}
Our analysis techniques build on the existing works outlined above, but also come with several new aspects and challenges; specific comparisons are deferred to Remarks \ref{rem:techniques} and \ref{rem:conv}.

\section{Main Results} \label{sec:main}

We first describe a randomized construction of a doubly-regular $T\times n$ test matrix $\mathsf{X}$, with $T=\frac{rn}{s}$, where $r$ and $s$ are variables to be chosen according to the setting being studied. We select the test matrix in the following manner:
\begin{itemize}
    \item Sample $r$ matrices $\mathsf{X}_1,\dots,\mathsf{X}_r$ independently, where each matrix $\mathsf{X}_j$ ($j\in\{1,\dots,r\}$) is sampled uniformly from all $\frac{n}{s}\times n$ binary matrices with exactly $s$ items per test (i.e., a row weight of $s$) and one item per column (i.e., a column weight of one).\footnote{We perform our analysis assuming that $\frac{n}{s}$ is an integer, since the effect of rounding is asymptotically negligible.}
    \item Form $\mathsf{X}$ by concatenating $\mathsf{X}_1,\dots,\mathsf{X}_r$ vertically.
\end{itemize}
In other words, we perform $r$ independent rounds of testing, where each round randomly partitions the items into $\frac{n}{s}$ tests of size $s$.  This approach was proposed in \cite{Bro20}, and was used as the first step of a two-stage procedure in \cite{Ald20}.  As we hinted in the previous section, it is distinct from the design that follows the uniform distribution over \emph{all} matrices with row weight $s$ and column weight $r$ (e.g., see \cite{Mez08}). The above design is considered primarily to facilitate the analysis; we expect the two designs to behave similarly, but we leave it as an open problem as to whether there exist settings in which one provably outperforms the other.

After testing the items using our test matrix, we run the DD algorithm, shown in Algorithm \ref{alg:COMP&DD_algo}, to attain our estimate $\widehat{\mathcal{K}}$ of the defective set.

\subsection{Unconstrained Testing in the Sub-Linear Regime}

We state our first main result as follows, and prove it in Section \ref{sec:uncon_sublinear}.

\begin{theorem} \label{thm:unconstrained_achievability}
Under the doubly-regular design described above with parameters $s = \frac{n \log 2}{k}$ and $r = c \log n$ for some constant $c > 0$,\footnote{This result holds regardless of whether we round these values up or down.} when there are $k=\Theta\big(n^\theta\big)$ defective items with constant $\theta\in(0,1)$, the DD algorithm attains vanishing error probability if
\begin{align}
    T\geq(1+\epsilon)\frac{\max\{\theta,1-\theta\}}{\log^2 2}k\log n,
\end{align}
where $\epsilon$ is an arbitrarily small positive constant (i.e., when the constant $c$ in $r = c \log n$ satisfies $c \geq(1+\epsilon)\frac{\max\{\theta,1-\theta\}}{\log2}$, noting that $T = \frac{rn}{s}$).
\end{theorem}

\subsection{Approximate Recovery in the Linear Regime}

We state our main result for the linear regime as follows, and prove it in Section \ref{sec:linear}. In Appendix \ref{app:COMP}, we also provide a counterpart of this result for the case that there are false negatives but no false positives.

\begin{theorem} \label{thm:linear_achievability}
Under the doubly-regular design described above with parameters $s$ and $r$, when there are $k=pn$ defective items with constant $p\in(0,1)$, the DD algorithm attains $\textup{FNR} \le \min\{1, \textup{FNR}_{\max}(1+o(1))\}$, where 
\begin{align}
    &\textup{FNR}_{\max} \nonumber \\
    &=\left((1-(1-p)^{s-1})+\frac{(1-p)(1-(1-p)^{s-1})^r}{p}\right)^r. \label{eq:linear_achievability}
\end{align}
The corresponding number of tests is $T = \frac{rn}{s}$.
\end{theorem}

Next, we pause momentarily to introduce the following definition to help us evaluate our result:
\begin{definition}[Rate] \label{def:rate}
Under the combinatorial prior with $n$ items, $k$ defectives and $T$ tests, the rate is equal to $\frac{\log_2{n \choose k}}{T}$, which measures the average number of bits of information that would be gained per test if $\mathcal{K}$ were recovered perfectly. In the linear regime with $k=pn$, this can be simplified to $\frac{nH_2(p)}{T}$ (up to a $1+o(1)$ equivalence), where $H_2(p)=-p\log_2{p}-(1-p)\log_2{(1-p)}$ is the binary entropy function.
\end{definition}

Returning to \eqref{eq:linear_achievability}, we observe that this expression is minimized when $r$ is large and $s$ is small. However, this conflicts with our goal of minimizing the number of tests $T=\frac{rn}{s}$, which implies that we should be aiming for small $r$ and large $s$ instead. To balance these conflicting goals, we evaluate our results in terms of their achievable rates, subject to a maximum permissible $\text{FNR}$.  We first partially do so analytically, and then turn to numerical evaluations.

The following corollary concerns the limit of a small proportion of defectives, i.e., the low-sparsity regime.

\begin{corollary} \label{cor:small_p}
    Under the setup of Theorem \ref{thm:linear_achievability}, there exist choices of $r$ and $s$ (depending on $p$) such that, in the limit as $p \to 0$ (after having taking $n \to \infty$), we have (i) $\mathrm{FNR}_{\max}$ approaches $0$, (ii) the rate approaches $\log 2$, and (iii) it holds that  $s=\frac{\log2}{p}(1+o(1))$ and $r = \frac{\log(\frac{1}{p})}{\log2}(1+o(1))$.
\end{corollary}
\begin{proof}
    We start by choosing $s=\frac{\log2}{p}$ (the effect of rounding is negligible as $p \to 0$). 
    By Theorem \ref{thm:linear_achievability}, we have
    \begin{align}
        &\text{FNR}_{\max} \nonumber \\
        &=\left(1-(1-p)^{s-1}+\frac{(1-p)(1-(1-p)^{s-1})^r}{p}\right)^r \\
        &\stackrel{(a)}{=}\bigg(1-e^{-p(s-1)(1+o(1))} + \frac{1}{p}\big(1-e^{-p(s-1)(1+o(1))}\big)^r \bigg)^r \\
        &\stackrel{(b)}{=}\bigg(\frac{1+o(1)}{2}+\frac{1}{p}\Big(\frac{1+o(1)}{2}\Big)^r\bigg)^r,
    \end{align}
    where (a) applies $(1-p)^{s-1}=e^{-p(s-1)(1+o(1))}$ as $p \to 0$, and (b) substitutes $s=\frac{\log2}{p}$. The above expression approaches zero if $\frac{1}{p}\big(\frac{1+o(1)}{2}\big)^r\leq\frac{1}{2}-\delta$ for some positive constant $\delta\in(0,0.5)$, because this gives $\text{FNR}_{\max}=(1-\delta + o(1))^{\omega(1)}=o(1)$. Making $r$ the subject, we obtain
    \begin{align}
        r&\geq\frac{\log p+\log(\frac{1}{2}-\delta)}{-\log 2 + o(1)} \nonumber \\
        &=\frac{\log(\frac{1}{p})-\log(\frac{1}{2}-\delta)}{\log2 + o(1)}
        =\frac{\log(\frac{1}{p})}{\log2}(1+o(1)). \label{eq:r_bound}
    \end{align}
    Hence, we have $r= \frac{ \log(\frac{1}{p})} {\log2}(1+o(1))$ as desired. Finally, 
    \begin{align}
        \text{rate}
        &=\frac{\log_2{n\choose k}}{T}
        \stackrel{(a)}{=}\frac{s\log_2{n\choose pn}}{rn} \nonumber \\
        &\stackrel{(b)}{=}\frac{sp\log(\frac{1}{p})}{r\log 2}(1+o(1))
        \stackrel{(c)}{=}(\log2)(1+o(1)),
    \end{align}
    where (a) substitutes $T=\frac{rn}{s}$, (b) applies $\log_2{n\choose pn}=\frac{1}{\log2} pn(1+o(1))\log\big(\frac{n}{pn}\big)$ for $p = o(1)$, and (c) substitutes $s=\frac{\log2}{p}$ and $r = \frac{\log(\frac{1}{p})}{\log2}(1+o(1))$ along with some simplifications.
\end{proof}

Corollary \ref{cor:small_p} is consistent with the fact that, in the sub-linear regime $k = o(n)$, DD can achieve a rate of $\log 2$ bits/test for an arbitrarily small target FNR \cite[Sec.~5.1]{Ald19}.  Essentially, these two results can both be viewed as taking the limits $\frac{k}{n} \to 0$ and $\mathrm{FNR} \to 0$, but in the opposite order.

\textbf{Numerical evaluation and comparison:} To numerically evaluate the result given in Theorem \ref{thm:linear_achievability}, we perform the following:
\begin{enumerate}
    \item Select a value $\alpha$ to be the maximum permissible $\text{FNR}$, and select values of $p$ from the interval $(0,0.5]$ to evaluate.
    \item For each $p$, numerically optimize the free parameters $(s,r)$ to minimize the \textit{aspect ratio} $\frac{T}{n}=\frac{r}{s}$, subject to $\text{FNR} \le \alpha$.
    \item Compute the rate $\frac{nH_2(p)}{T}$, and plot this over the chosen values of $p$.
\end{enumerate}

The rates attained by the doubly-regular design (from Theorem \ref{thm:linear_achievability}) are shown in Figure \ref{fig:dconstrate}.

\begin{figure}
    \centering
    \includegraphics[width=0.475\textwidth]{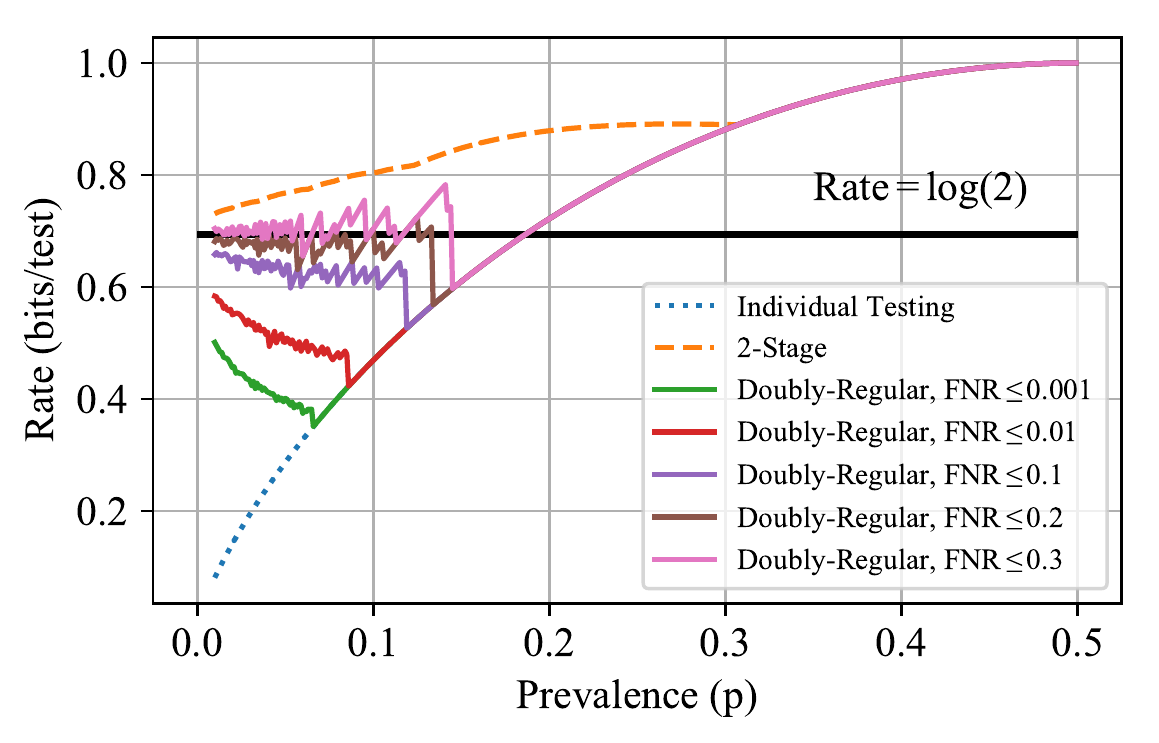}
    \caption{Achievable rates for DD decoding with the doubly-regular design and approximate recovery (along with individual testing and a two-stage design \cite{Ald20}, both of which attain exact recovery).} \vspace*{-1.5ex}
    \label{fig:dconstrate}
\end{figure}

From Figure \ref{fig:dconstrate}, we observe that doubly-regular testing with the DD algorithm attains strictly higher rates than individual testing for smaller values of $p$, while reducing to individual testing for larger $p$. The extent of improvement increases as the target FNR increases. However, the rates obtained by conservative 2-stage testing (with exact recovery) remain higher.  We are not aware of analogous results on the FNR for other non-adaptive designs such as Bernoulli or near-constant column weight, but in Appendix \ref{app:COMP} we compare to those designs under COMP decoding using the FPR instead of FNR, and see that the doubly-regular design almost always outperforms them.

The discontinuities in the plot can be explained by the fact that $r$ (tests per item) and $s$ (items per test) must both be integers. Since $\text{FNR}_{\max}$ is increasing in $p$, the pair $(r,s)$ will change whenever $\text{FNR}_{\max}$ exceeds $\alpha$. This leads to a downward jump in rate, albeit with $\text{FNR}_{\max}$ potentially being significantly smaller than $\alpha$.

\begin{figure}
    \includegraphics[width=0.45\textwidth]{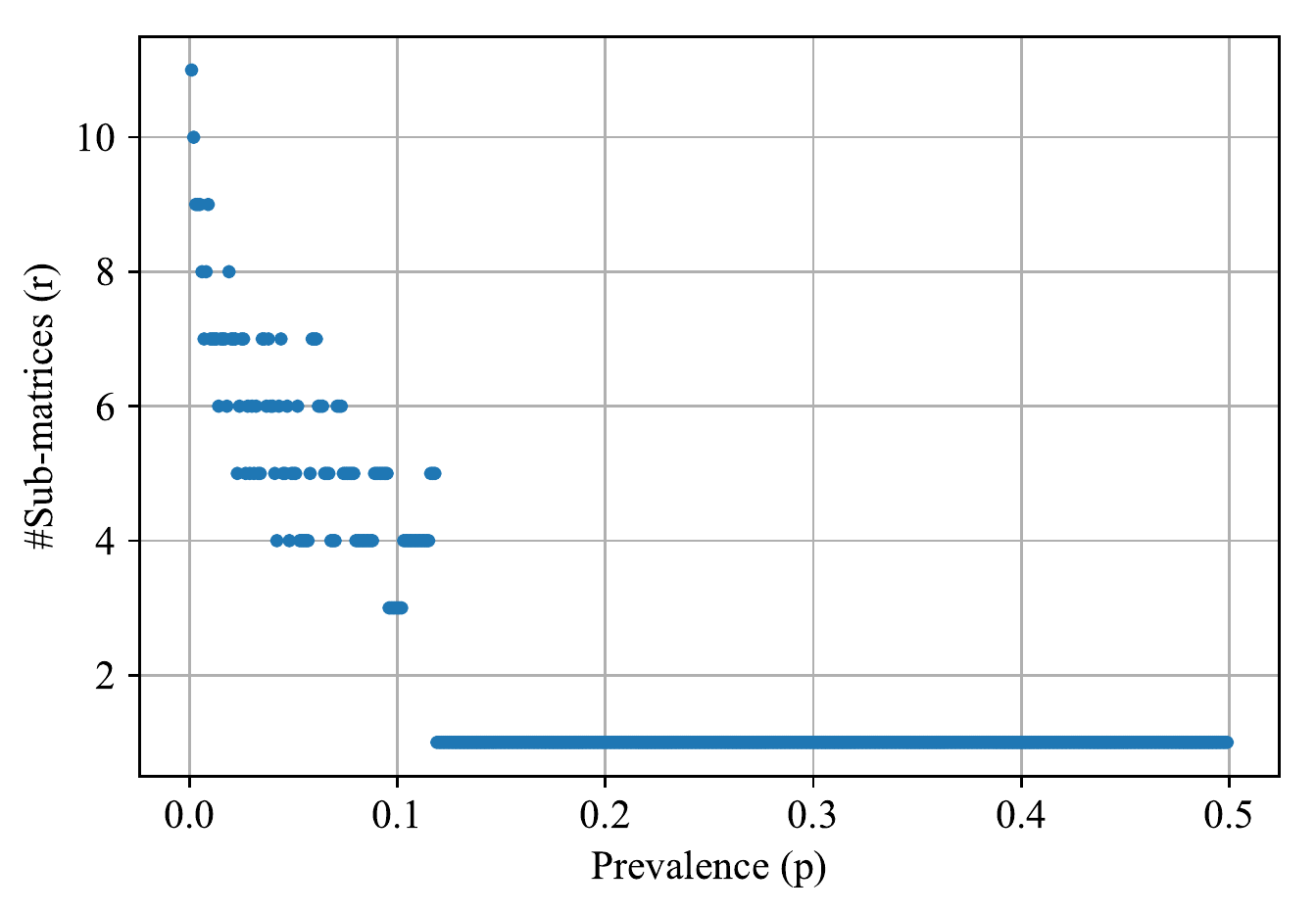} \quad
     \includegraphics[width=0.45\textwidth]{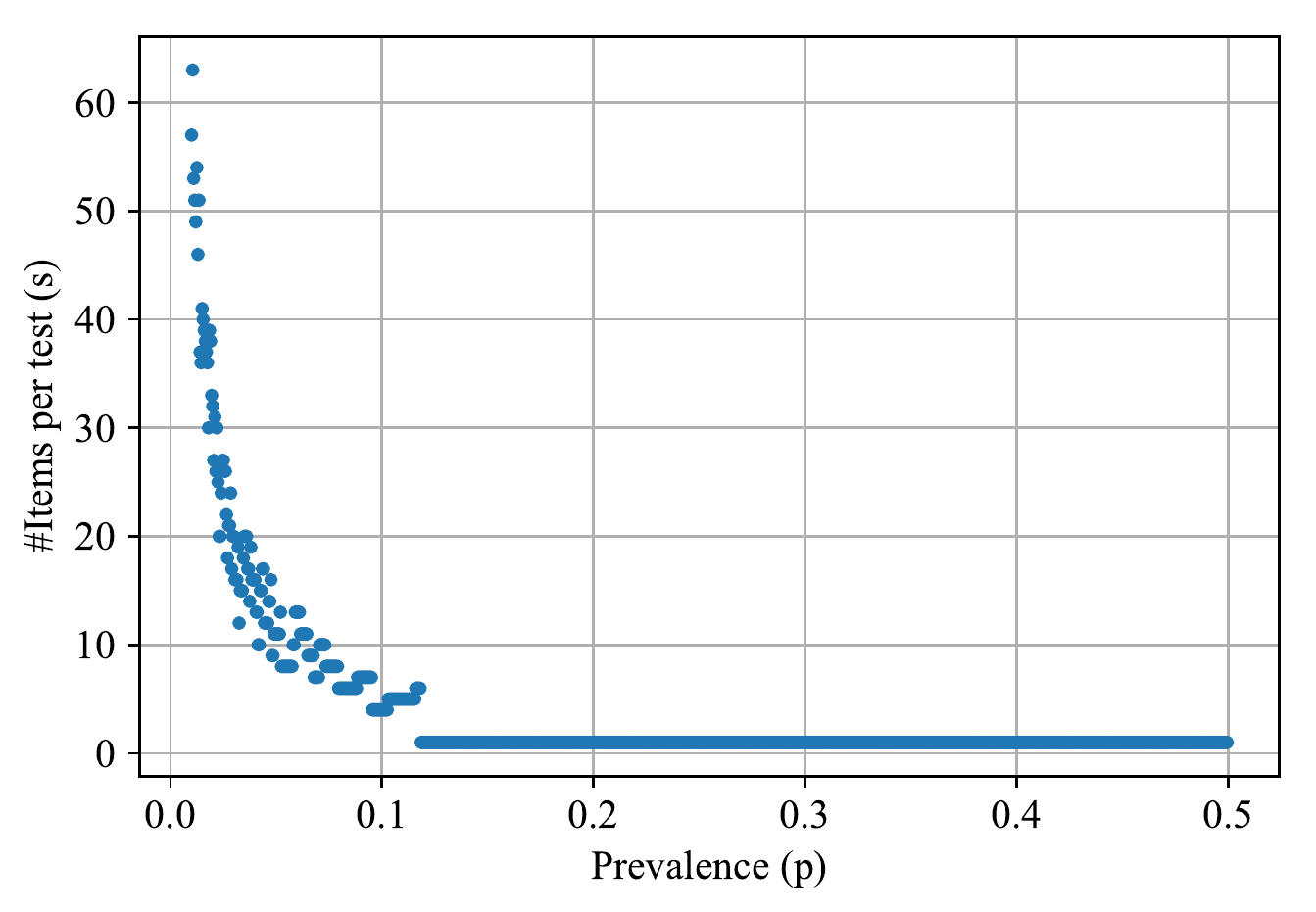}     

    \caption{Optimal $r$ (Left) and $s$ (Right) for $\alpha = 0.1$.} \vspace*{-1.5ex}
    \label{fig:DConstRS}
\end{figure}

Figure \ref{fig:DConstRS} illustrates the optimal values of $r$ and $s$. Generally, $r$ is small, since $\text{FNR}_{\max}$ decreases exponentially as $r$ increases. Although the curve for $s$ is not smooth, it empirically satisfies $s=\Theta(p^{-1})$, which is consistent with Corollary \ref{cor:small_p}.  

At this stage, one may wonder why some of the curves in Figure \ref{fig:dconstrate} appear to approach a value strictly smaller than $\log 2 \approx 0.693$ despite Corollary \ref{cor:small_p}.  The reason is that this behavior is only observed for {\em extremely} small $p$, and the rate almost instantaneously drops to a significantly smaller value as $p$ increases.  In fact, we found that if we set $s = \frac{\log 2}{p}$ and $r = \frac{\log(\frac{1}{p})}{\log2}$ as suggested by Corollary \ref{cor:small_p}, our upper bound on the FNR exceeds one even when $p$ is brought down to the order of $10^{-14}$.  These findings suggest that asymptotic results for the regime $k = o(n)$ should be interpreted with caution when it comes to practical problem sizes.

{\bf Discussion on possible converse results.} It is difficult to gauge the tightness of our achievability result, due to the lack of converse results in this setting.  While we do not attempt to make any formal statements addressing this, we believe that the constant $\log 2$ in Corollary \ref{cor:small_p} is likely to be the best possible.  This is supported by the following:
\begin{itemize}
    \item Under the near-constant tests-per-item design, the DD algorithm is known to fail to attain ${\rm FNR} \to 0$ at rates exceeding $\log 2$ bits/test \cite{Coj19}.  Furthermore, it appears unlikely that any algorithm could outperform DD for the goal of attaining both ${\rm FPR} = 0$ and ${\rm FNR} \to 0$.\footnote{On the other hand, an improved rate of $1$ bit/test is possible when we only require ${\rm FPR} \to 0$ and ${\rm FNR} \to 0$ \cite{Sca15b}.}
    \item In Appendix \ref{app:COMP}, we study the ``opposite'' goal of attaining ${\rm FNR} = 0$ and ${\rm FPR} \to 0$, and in that case one can rigorously show that $\log 2$ bits/test is the best possible.
\end{itemize} 

\subsection{Size-Constrained Sub-Linear Regime}

We state our achievability result below, and prove it in Section \ref{sec:con_sublinear}.

\begin{theorem} \label{thm:sparse_unconstrained_achievability}
For $k=\Theta(n^{\theta})$ with $\theta\in[0,1)$, and $\rho=\Theta\big(\big(\frac{n}{k}\big)^\beta\big)$ with $\beta\in(0,1)$, for any integer $r$ satisfying:
\begin{itemize}
    \item If $\theta\geq\frac{1}{2}$: $r>\frac{\theta}{(1-\theta)(1-\beta)}$ and $r\geq\frac{2-\beta}{1-\beta}$;
    \item If $\theta<\frac{1}{2}$: $r\geq\frac{1-\theta\beta}{(1-\theta)(1-\beta)}$;
\end{itemize}
the DD algorithm with $T = \frac{rn}{\rho}$ tests, chosen according to the above randomized doubly-regular design with $s = \rho$, recovers the defective set with asymptotically vanishing error probability.
\end{theorem}

We additionally provide a converse result, which is stated as follows and proved in Section \ref{sec:uncon_sublinear_converse}.

\begin{theorem} \label{thm:rho_converse}
Suppose that $k=\Theta\big(n^\theta\big)$, for $\theta\in(0,1)$, and let $\mathsf{X}$ be a non-adaptive test matrix such that each test contains at most $\rho=\Theta\big(\big(\frac{n}{k}\big)\big)^\beta$ items, for $\beta\in(0,1)$. Given\footnote{The $\frac{1}{1-\beta}$ term comes from the converse in \cite{Ven19}, and need not be integer-valued.}
\begin{align}
    r&=\max\bigg\{2,\frac{1}{1-\beta},\bigg\lceil\frac{1-(1-\theta)(2\beta+1)}{(1-\theta)(1-\beta)}\bigg\rceil\bigg\}, \label{eq:two_cases_previously}
\end{align}
for an arbitrary constant $\epsilon>0$, if there are $(1-\epsilon)\frac{rn}{\rho}$ or fewer tests, then any decoder has error probability $1-o(1)$.
\end{theorem}

The plots of the constant $r$ against $\theta$ are displayed in Figure \ref{fig:c_vs_theta_plots}. Our main results do not apply to the case that $\beta = 0$ exactly (which was already handled in \cite{Oli20}), but we can plot the relevant limits as $\beta \to 0$.  Similarly, the results of \cite{Oli20} do not apply when $\beta > 0$, but we can plot the relevant limits as $\rho \to \infty$.  In Figure \ref{fig:c_vs_theta_plots} (top-left), the same curve is obtained in both limits, and in both cases we have matching achievability and converse bounds. 

For the other values of $\beta$ shown, we observe strict improvements of DD over COMP, and the gap widens as $\beta$ increases. For the converse, we similarly observe a strict improvement over the previous converse for $\beta\in(0,1)$.

\begin{figure*}
    \begin{tabular}{ll}
    \includegraphics[width=0.45\textwidth]{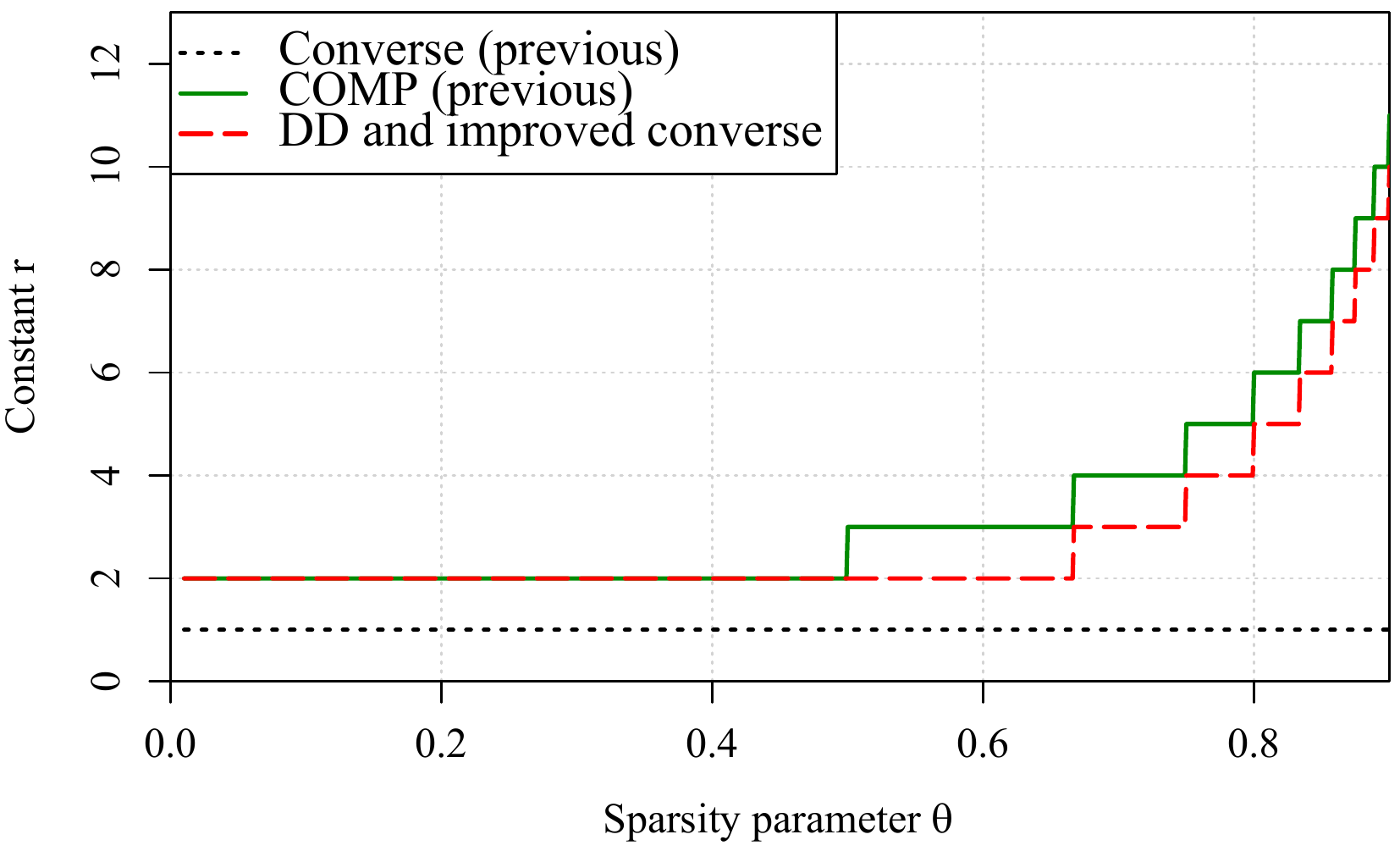} &
    \includegraphics[width=0.45\textwidth]{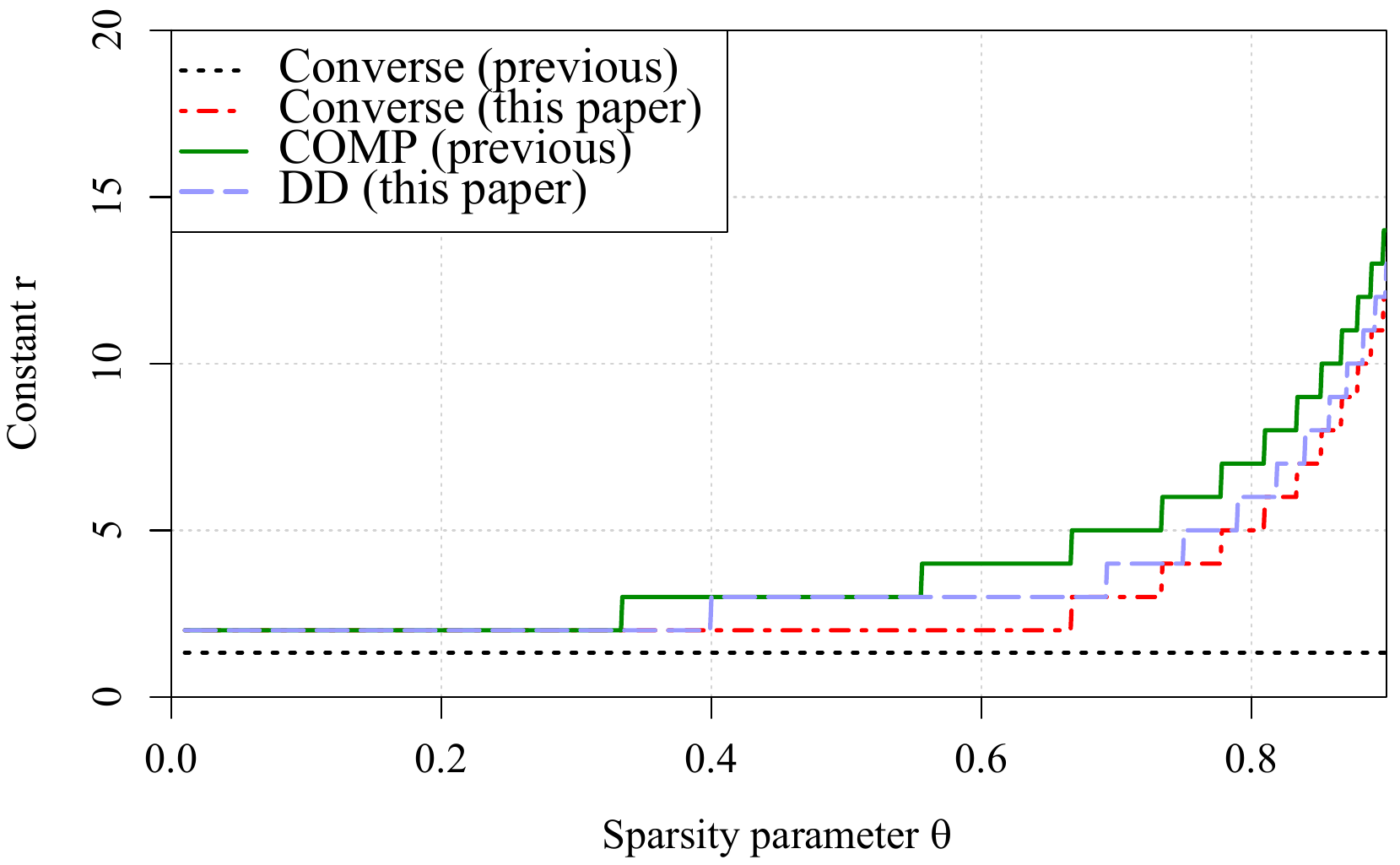} \\
    \includegraphics[width=0.45\textwidth]{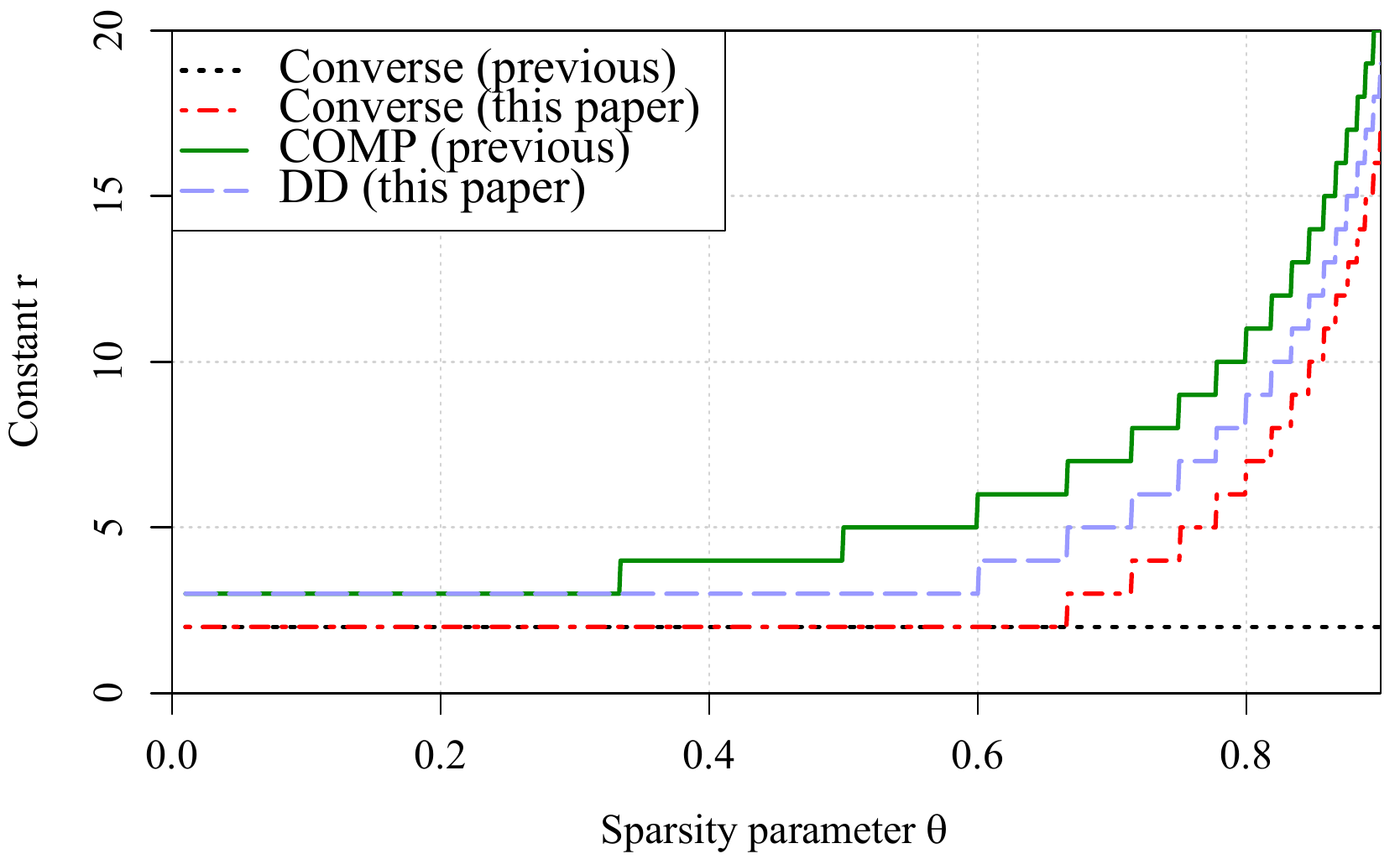} &
    \includegraphics[width=0.45\textwidth]{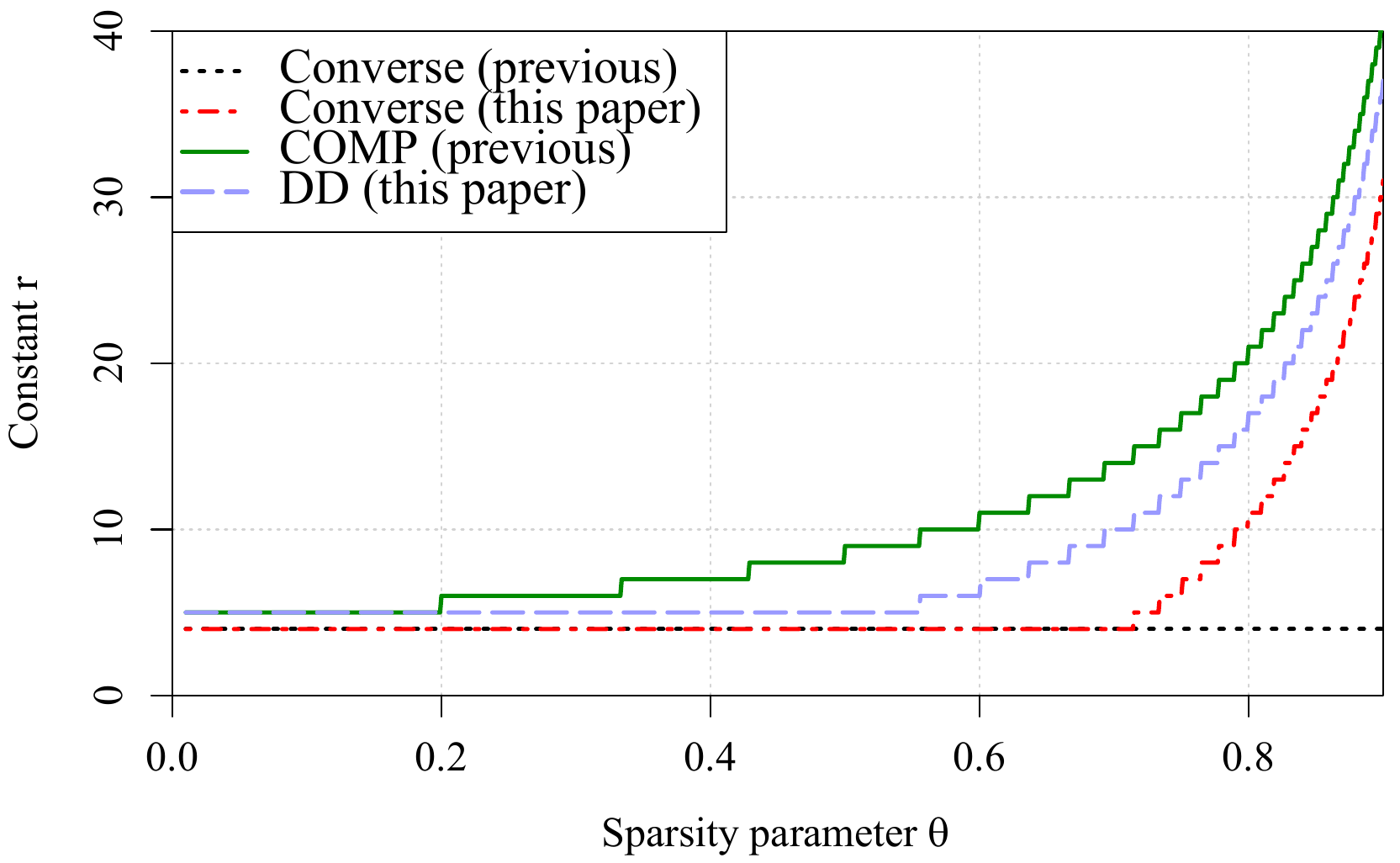}
    \end{tabular}
    \caption{Plots of $r$ vs. $\theta$ for $\beta \to 0$ (top-left), $\beta=0.25$ (top-right), $\beta=0.5$ (bottom-left), and $\beta=0.75$ (bottom-right).} \vspace*{-1.5ex}
    \label{fig:c_vs_theta_plots}
\end{figure*}

\section{Achievability Analysis} \label{sec:ach_ana}

In this section, we prove the three achievability results stated above.   In general, doubly-regular designs come with more complicated dependencies that are difficult to handle. The construction that we consider (i.e., concatenating independent doubly-regular sub-matrices with column weight one) allows us to simplify the analysis of the DD algorithm by extending the analysis of one sub-matrix to the entire test matrix.

We proceed by outlining the key steps of the analysis and introducing the relevant notation. This applies to all three settings that we consider, and the differences lie in the specific details (e.g., the choices of $r$ and $s$) and the subsequent parts that extend from these steps. The key steps are:
\begin{enumerate}
    \item \textbf{Determine concentration of \#non-defective items in $\mathcal{PD}$:} Let $\mathcal{G}$ be the set of non-defective items in $\mathcal{PD}$, and $G=|\mathcal{G}|$. Furthermore, for each non-defective $i$, let $\mathrm{PD}_i=\mathds{1}\{i\in\mathcal{PD}\}$. Then, $\sum_{i\in[n]\setminus\mathcal{K}}\mathrm{PD}_i=G$. This step concerns the concentration of $G$ around its mean. We start by considering the probability that a particular non-defective item $i$ is found in a negative test in a single sub-matrix $\textsf{X}_j$. Consider the unique test which item $i$ is found in. The probability that all the other $s-1$ items in the test are non-defective is
    \begin{align}
        \frac{{n-k-1 \choose s-1}}{{n-1 \choose s-1}} 
        &=\frac{(n-k-1)!(n-s)!}{(n-1)!(n-k-s)!} \nonumber \\
        &=\prod_{i=1}^{s-1}\frac{n-i-k}{n-i}
        =\prod_{i=1}^{s-1}\Big(1-\frac{k}{n-i}\Big). \label{eq:prob_nondef_in_neg_test_general}
    \end{align}
    Hence, consider all of the $r$ sub-matrices, we obtain
    \begin{align}
        \E[\mathrm{PD}_i]
        &\stackrel{(a)}{=}\mathbb{P}[i\in\mathcal{PD}] \nonumber \\
        &\stackrel{(b)}{=}\bigg(1-\prod_{i=1}^{s-1}\Big(1-\frac{k}{n-i}\Big)\bigg)^{r}, \label{eq:E[PD_i]_general}
    \end{align}
    where (a) applies $\E[{\mathds{1}\{A\}}]=\mathbb{P}[A]$, and (b) applies \eqref{eq:prob_nondef_in_neg_test_general}. Hence, we have
    \begin{align}
        \E[G]
        &=(n-k)\E[\mathrm{PD}_i] \nonumber \\
        &=(n-k)\bigg(1-\prod_{i=1}^{s-1}\Big(1-\frac{k}{n-i}\Big)\bigg)^{r}. \label{eq:E[G]_general}
    \end{align}
    Depending on the setting, we use an appropriate concentration inequality (e.g., Chebyshev's inequality) to show that $G$ is close to this value with probability $1-o(1)$.
    
    \item \textbf{Determine concentration of \#masked defective items:} We introduce two further definitions.
    \begin{definition}
    Consider a defective item $i$ and a set $\mathcal{L}$ that does not include $i$. We say that defective item $i$ is \textit{masked} by $\mathcal{L}$ if every test that includes $i$ also includes at least one item from $\mathcal{L}$.
    \end{definition}
    \begin{definition} \label{def:masked_def_item}
    We call an item $i$ \textit{masked} if it is masked by $\mathcal{K}\setminus\{i\}$.\footnote{Here we are concerned with the case that $i$ is defective, but later we will use this definition where $i$ is non-defective, and hence $\mathcal{K}\setminus\{i\} = \mathcal{K}$.}  If the matrix is not specified then this property is defined with respect to the full test matrix $\mathsf{X}$, but we will also use the same terminology with respect to a given sub-matrix $\mathsf{X}_j$.
    \end{definition}
    For a given defective item $i$, let $\mathrm{M}_i^j=\mathds{1}\{\text{$i$ is masked in $\mathsf{X}_j$}\}$, and let $M^j=\sum_{i\in\mathcal{K}}\mathrm{M}_i^j$ be the number of masked defective items in $\mathsf{X}_j$. Without loss of generality (WLOG), suppose that items $1,\dots,k$ are defective. We have
    \begin{align}
        \mathbb{P}\big[\mathrm{M}_i^j=0\big]
        =\frac{{n-k\choose s-1}}{{n-1\choose s-1}}
        =\prod_{i=1}^{s-1}\Big(1-\frac{k-1}{n-i}\Big), \label{eq:P[M_i^j]_general}
    \end{align}
    where the last equality uses the same steps as \eqref{eq:prob_nondef_in_neg_test_general}. Summing over all $k$ defective items, it follows that
    \begin{align}
        \E\big[M^j\big]
        &=k\big(1-\mathbb{P}\big[\mathrm{M}_i^j=0\big]\big) \nonumber \\
        &=k\bigg(1-\prod_{i=1}^{s-1}\Big(1-\frac{k-1}{n-i}\Big)\bigg). \label{eq:E[M_i^j]_general}
    \end{align}
    Depending on the setting, we use an appropriate technique to show that $M^j$ concentrates around this value for all $\mathsf{X}_j$ simultaneously, with probability $1-o(1)$.
    
    \item \textbf{Establish conditional independence:} In this step, we condition on the preceding high-probability bounds on $G$ and $M^j$ holding.  To facilitate the analysis, it is useful to not only condition on such events, but to condition on more specific events that ensure such concentration (apart from these, one final conditioning event will also be given below):
    \begin{itemize}
        \item We condition on a fixed set of tests being positive, and the remaining tests being negative.  This fixed set has no explicit constraints.
        \item We condition on a fixed realization of the defective set $\mathcal{K}$.
        \item We condition on $\mathcal{G}$ (the non-defectives that are marked as possibly defective) being a fixed set with some fixed size $G$.  This value of $G$ is assumed to satisfy the above-established concentration behavior.
        \item For $j=1,\dotsc,r$, we condition on the {\em number} of masked defectives in sub-matrix $\mathsf{X}_j$ being $M^j$, whose values again satisfy our established concentration results.  Note that unlike with $\mathcal{G}$, we do not condition on the specific set of masked item indices, but instead only on the total number.
    \end{itemize}
    Once we show that the conditional error probability is suitably small, it follows that the same is true after averaging (over all realizations of positive tests, all possible $\mathcal{K}$, and so on).

    By the symmetry of the test design, without loss of generality, we can consider the following:
    \begin{itemize}
        \item The defective items are indexed by $1,\dotsc,k$.
        \item The items in $\mathcal{G}$ indexed by $k+1,k+2,\dotsc,k+G$.
    \end{itemize}
    A slight issue here is that if we naively condition on the above sets of size $G$ and $M^j$, the resulting sub-matrices $\mathsf{X}_1,\dots,\mathsf{X}_r$ will typically not be conditionally independent, since the sub-matrices are coupled via $\mathcal{G}$.  Specifically, conditioning on $\mathcal{G}$ amounts to conditioning on (i) items $k+1,k+2,\dots,k+G$ only being in positive tests in {\em all} sub-matrices, and (ii) each of items $k+G+1,k+G+2,\dots,n$ being in a negative test in {\em some} sub-matrix.  The latter of these means, for example, that if we are told that item $k+G+1$ is in a negative test in submatrix 1, it reduces the probability of the same being true in submatrix 2 (thus violating independence).
    
    
    To overcome this difficulty, we additionally condition on the {\em fixed and specific} placements of items $k+G+1,k+G+2,\dots,n$ into negative tests (but not positive tests).  This must be done subject to each of them being in some negative test, but the placements are otherwise arbitrary and do not impact our analysis.  
    
    When this additional conditioning is done, the overall conditioning involving $G$ and $M^j$ becomes an ``AND'' of $r$ sub-events (one per sub-matrix), with each sub-event requiring that (i) $M^j$ defectives are masked in  sub-matrix $j$, (ii) items $k+1,k+2,\dots,k+G$ are only included in positive tests, and (iii) items $k+G+1,k+G+2,\dots,n$ are included in the specific negative tests indicated.  Due to this ``AND'' structure, conditional independence is maintained among the $r$ sub-matrices.  For completeness, we provide the mathematical details of this finding in Appendix \ref{app:condition}.
    
    \item \textbf{Bound the total error probability:} Recall that $\mathrm{M}_i^j$ is the indicator event that defective $i$ is masked by $\mathcal{K}\setminus\{i\}$ in $\mathsf{X}_j$, and let $\mathrm{MG}_i^j$ be the event that defective item $i$ is masked by $\mathcal{G}$ in $\mathsf{X}_j$. For a given $\mathsf{X}_j$, we have the following, where $\mathcal{E}$ is the event described in step 3 above and the subscript $\mathbb{P}_{\mathcal{E}}[\cdot]$ indicates conditioning:
    \begin{align}
        &\mathbb{P}_{\mathcal{E}}\big[\mathrm{M}_i^j\cup \mathrm{MG}_i^j\big] \nonumber \\
        &=\mathbb{P}_{\mathcal{E}}\big[\mathrm{M}_i^j\big]+\mathbb{P}_{\mathcal{E}}\big[\mathrm{MG}_i^j\cap\neg \mathrm{M}_i^j\big] \\
        &=\mathbb{P}_{\mathcal{E}}\big[\mathrm{M}_i^j\big]+\mathbb{P}_{\mathcal{E}}\big[\neg \mathrm{M}_i^j\big]\mathbb{P}_{\mathcal{E}}\big[\mathrm{MG}_i^j|\neg \mathrm{M}_i^j\big] \\
        &\stackrel{(a)}{=}\frac{M^j}{k}+\Big(1-\frac{M^j}{k}\Big)\mathbb{P}_{\mathcal{E}}\big[\mathrm{MG}_i^j|\neg \mathrm{M}_i^j\big] \\
        &\stackrel{(b)}{\leq}\frac{M^j}{k}+\Big(1-\frac{M^j}{k}\Big)\frac{G}{k-M^j} \label{eq:loose_upp_bound} \\
        &=\frac{M^j}{k}+\frac{G}{k}, \label{eq:M_G_k}
    \end{align}
    where:
    \begin{itemize}
        \item (a) uses the fact that that we conditioned on having $M^j$ masked defective items, so by the symmetry of the randomized test design, we have $\mathbb{P}\big[\mathrm{M}_i^j\big]=\frac{M^j}{k}$.
        \item (b) holds because conditioned on $\neg \mathrm{M}_i^j$, item $i$ is in one of the $k-M^j$ tests with a single defective item (recall that $k-M^j$ is the number of non-masked defective items). Again using symmetry, each non-defective item has at most $\frac{1}{k-M^j}$ probability of being in that particular test (due to there existing $k - M^j$ equally likely options). Taking the union bound over all items in $\mathcal{G}$ gives the desired bound.
    \end{itemize}
    By \eqref{eq:M_G_k} and the conditional independence of the $\mathsf{X}_j$'s, the probability of defective item $i$ being masked by $\mathcal{PD}\setminus\{i\}$ in all sub-matrices satisfies
    \begin{align}
        \prod_{j=1}^{r}\mathbb{P}_{\mathcal{E}}\big[\mathrm{M}_i^j\cup \mathrm{MG}_i^j\big]
        &\leq\Big(\frac{M^j}{k}+\frac{G}{k}\Big)^r. \label{eq:prob_def_i_masked_by_PD_general}
    \end{align}
\end{enumerate}

\begin{remark} \label{rem:techniques}
    {\em (Comparison to existing techniques)} Our approach is distinct from existing analyses of the DD algorithm \cite{Ald14a,Joh16,Oli20}, which use a ``globally symmetric'' random matrix with no sub-matrix block structure.  Like all of these works, we still adopt the the high-level steps of characterizing $G$ and then characterizing $M^j$ conditioned on $G$, but the details are largely different, including the unique aspect of transferring results regarding simpler sub-matrices to the entire test matrix.  Moreover, regarding the size-constrained regime, we note that the analysis of \cite{Oli20} relies strongly on the assumption $\rho = O(1)$, which is what led us to adopting a distinct approach.  Compared to \cite{Oli20}, we believe that our test design and analysis are relatively simpler, though the regimes handled are different and neither subsumes the other ($\beta = 0$ vs.~$\beta > 0$).
\end{remark}

\subsection{Unconstrained Sub-Linear Regime With Exact Recovery (Theorem \ref{thm:unconstrained_achievability})} \label{sec:uncon_sublinear}

We follow the four steps described above to obtain our result, and begin by selecting the relevant parameters. We choose $s=\frac{n\log2}{k}$ and $r=c\log n$, where $c$ is a constant to be specified shortly\footnote{Rounding $r$ has a negligible effect on $c$ (only changing by a factor of $1+o(1)$) and is thus ignored in our analysis.}. This results in each $\mathsf{X}_j$ being a $\frac{k}{\log2}\times n$ sub-matrix, and $\mathsf{X}$ being a $\frac{ck\log n}{\log 2}\times n$ matrix.

\noindent\textbf{Step 1:} Setting $s=\frac{n\log2}{k}$ and $r=c\log n$ in \eqref{eq:E[G]_general}, we have
\begin{align}
    &\E[G] \nonumber \\
    &=\Theta\bigg(n\bigg(1-\prod_{i=1}^{(n/k)\log2-1}\Big(1-\frac{k}{n-i}\Big)\bigg)^{c\log n}\bigg) \label{eq:start_binom_frac_simplification} \\
    &\stackrel{(a)}{=}\Theta\bigg(n\bigg(1-\Big(1-\frac{k}{n(1-o(1))}\Big)^{(n/k)\log2-1}\bigg)^{c\log n}\bigg) \\
    &\stackrel{(b)}{=}\Theta\bigg(n\bigg(1-\exp\Big(-\frac{k}{n}\frac{n\log2}{k}(1+o(1))\Big)\bigg)^{c\log n}\bigg) \\
    &=\Theta\bigg(n\Big(\frac{1}{2}(1+o(1))\Big)^{c\log n}\bigg) \label{eq:end_binom_frac_simplification} \\
    &=\Theta\big(n^{1-(c\log2)(1+o(1))}\big), \label{eq:E[G]_uncon_sublinear}
\end{align}
where (a) uses $i\leq\frac{n\log2}{k}-1=o(n)$, and (b) uses the fact $(1+a)^b=e^{ab(1+o(1))}$ when $|a|=o(1)$. Applying Markov's inequality, we have
\begin{align}
    \mathbb{P}\Big[G\geq\frac{k}{\log n}\Big]
    &\le\frac{\E[G]}{k/\log n}
    \stackrel{(a)}{=}\frac{\Theta(n^{1-(c\log2)(1+o(1))})}{k^{1-o(1)}} \nonumber \\
    &\stackrel{(b)}{=}\Theta\big(n^{1-c\log2-\theta+o(1)}\big), \label{eq:G_prob_bound}
\end{align}
where (a) uses \eqref{eq:E[G]_uncon_sublinear} and $\frac{k}{\log n}=k^{1-o(1)}$, and (b) substitutes $k=\Theta\big(n^\theta\big)$ and further simplifies the expression. Observe that the power of $n$ is below zero when $c\geq(1+\epsilon)\frac{1-\theta}{\log2}$, where $\epsilon$ is any positive constant. This implies that $G=o(k)$ with probability $1-o(1)$ when $c\geq(1+\epsilon)\frac{1-\theta}{\log2}$.

\noindent\textbf{Step 2:} Setting $s=\frac{n\log2}{k}$ and $r=c\log n$ in \eqref{eq:E[M_i^j]_general}, we obtain
\begin{align}
    \E\big[M^j\big]
    &=k\bigg(1-\prod_{i=1}^{(n/k)\log2-1}\Big(1-\frac{k-1}{n-i}\Big)\bigg)
    =\frac{k}{2}(1+o(1)), \label{eq:avg_M}
\end{align}
where the last equality uses the same steps as those in \eqref{eq:start_binom_frac_simplification}--\eqref{eq:end_binom_frac_simplification} to simplify the product to $\frac{1}{2}(1+o(1))$. Next we have
\begin{align}
    \Var\big[M^j\big]
    &=\Var\bigg[\sum_{i=1}^k\mathrm{M}_i^j\bigg] \nonumber \\
    &=k\Var\big[\mathrm{M}_i^j\big]+k(k-1)\Cov\big[M_1^j,M_2^j\big]. \label{eq:var_M^j}
\end{align}
We proceed to evaluate the variance and covariance terms separately. The calculation is the same for any value of $i$, so we fix $i=1$ and write
\begin{align}
    \Var\big[M_1^j\big]
    &=\E\big[(M_1^j)^2\big]-\E[M_1^j]^2 \nonumber \\
    &=\mathbb{P}\big[M_1^j=1\big]-\mathbb{P}[M_1^j=1]^2 \label{eq:var_of_M_i^j} \nonumber \\
    &=\frac{1}{4}(1+o(1)),
\end{align}
since the equality $\mathbb{P}\big[M_1^j=1\big] = \frac{1}{2}(1+o(1))$ is implicit in \eqref{eq:avg_M}.  Next, we define $S_{12}^j$ to be the event that defective items $1$ and $2$ are in the same test in $\mathsf{X}_j$, yielding
\begin{align}
    \mathbb{P}\big[S_{12}^j\big]
    =\frac{{n-2\choose\frac{n\log2}{k}-2}}{{n-1\choose\frac{n\log2}{k}-1}}
    =\frac{(n/k)\log2-1}{n-1}
    \leq\frac{\log2}{k}. \label{eq:S12j_uncon_sublinear}
\end{align}
In addition, we have
\begin{align}
    & \Cov\big[M_1^j,M_2^j\big] \nonumber \\
    &=\E\big[M_1^j,M_2^j\big]-\E\big[M_1^j\big]\E\big[M_2^j\big] \nonumber \\
    &=\mathbb{P}\big[M_1^j=1,M_2^j=1\big]-\mathbb{P}\big[M_1^j=1\big]\mathbb{P}\big[M_2^j=1\big]. \label{eq:cov_M1j_and_M2j}
\end{align}
\begin{figure*}[t]
    \normalsize
    \setcounter{mytempeqncnt}{\value{equation}}
    \begin{align}
        \mathbb{P}\big[M_1^j=1,M_2^j=1\big]
        &=\mathbb{P}\big[S_{12}^j\big]\mathbb{P}\big[M_1^j=1,M_2^j=1|S_{12}^j\big]+\mathbb{P}\big[\neg S_{12}^j\big]\mathbb{P}\big[M_1^j=1,M_2^j=1|\neg S_{12}^j\big] \label{eq:first_in_long} \\
        &\stackrel{(a)}{\leq}\frac{\log2}{k}+\mathbb{P}\big[\neg S_{12}^j\big]\mathbb{P}\big[M_1^j=1,M_2^j=1|\neg S_{12}^j\big] \\
        &\stackrel{(b)}{\leq}\frac{\log2}{k}+1-\mathbb{P}\big[M_1^j=0\cup M_2^j=0|\neg S_{12}^j\big] \\
        &\stackrel{(c)}{\leq}\frac{\log2}{k}+1-\mathbb{P}\big[M_1^j=0|\neg S_{12}^j\big]-\mathbb{P}\big[M_2^j=0|\neg S_{12}^j\big]+\mathbb{P}\big[M_1^j=0,M_2^j=0|\neg S_{12}^j\big] \\
        &\stackrel{(d)}{=}\frac{\log2}{k}+1-2\frac{{n-k\choose\frac{n\log2}{k}-1}}{{n-2\choose\frac{n\log2}{k}-1}}
        +\frac{{n-k\choose\frac{n\log2}{k}-1}{n-k-\frac{n\log2}{k}+1\choose\frac{n\log2}{k}-1}}{{n-2\choose\frac{n\log2}{k}-1}{n-\frac{n\log2}{k}-1\choose\frac{n\log2}{k}-1}}, \label{eq:prob_M1j_and_M_2j_both_1}
    \end{align}
    \hrulefill 
    \begin{align}
        \Cov\big[M_1^j,M_2^j\big]
        &\leq\frac{\log2}{k}+1-2\frac{{n-k\choose\frac{n\log2}{k}-1}}{{n-2\choose\frac{n\log2}{k}-1}}
        +\frac{{n-k\choose\frac{n\log2}{k}-1}{n-k-\frac{n\log2}{k}+1\choose\frac{n\log2}{k}-1}}{{n-2\choose\frac{n\log2}{k}-1}{n-\frac{n\log2}{k}-1\choose\frac{n\log2}{k}-1}}
        -\mathbb{P}\big[M_1^j=1\big]\mathbb{P}\big[M_2^j=1\big] \label{eq:first_in_long_2} \\
        &\stackrel{(a)}{=}\frac{\log2}{k}+1-2\frac{{n-k\choose\frac{n\log2}{k}-1}}{{n-2\choose\frac{n\log2}{k}-1}}
        +\frac{{n-k\choose\frac{n\log2}{k}-1}{n-k-\frac{n\log2}{k}+1\choose\frac{n\log2}{k}-1}}{{n-2\choose\frac{n\log2}{k}-1}{n-\frac{n\log2}{k}-1\choose\frac{n\log2}{k}-1}}
        -\bigg(1-\frac{{n-k\choose\frac{n\log2}{k}-1}}{{n-1\choose\frac{n\log2}{k}-1}}\bigg)^2 \\
        &\stackrel{(b)}{=}\frac{\log2}{k}
        +2\underbrace{\Bigg(\frac{{n-k\choose\frac{n\log2}{k}-1}}{{n-1\choose\frac{n\log2}{k}-1}}-\frac{{n-k\choose\frac{n\log2}{k}-1}}{{n-2\choose\frac{n\log2}{k}-1}}\Bigg)}
        _{\text{first part}}
        +\underbrace{\frac{{n-k\choose\frac{n\log2}{k}-1}{n-k-\frac{n\log2}{k}+1\choose\frac{n\log2}{k}-1}}{{n-2\choose\frac{n\log2}{k}-1}{n-\frac{n\log2}{k}-1\choose\frac{n\log2}{k}-1}}
            -\frac{{n-k\choose\frac{n\log2}{k}-1}^2}{{n-1\choose\frac{n\log2}{k}-1}^2}}
        _{\text{second part}}, \label{eq:cov_M1j_M2j}
    \end{align}
    \hrulefill
    \vspace*{4pt}
\end{figure*}
\noindent Focusing on the first term, by the law of total probability, we have \eqref{eq:first_in_long}--\eqref{eq:prob_M1j_and_M_2j_both_1} at the top of the next page, where (a) substitutes \eqref{eq:S12j_uncon_sublinear} and uses $\mathbb{P}\big[M_1^j=1,M_2^j=1|S_{12}^j\big]\le 1$, 
(b) uses $\mathbb{P}\big[\neg S_{12}^j\big]\leq1$ and $\mathbb{P}[A\cap B]=1-\mathbb{P}[\neg A\cup\neg B]$ for any two events $A$ and $B$, (c) uses the inclusion-exclusion principle, and (d) calculates $\mathbb{P}\big[M_1^j=0|\neg S_{12}^j\big]$ by counting the number of ways to choose the other $\frac{n\log2}{k}-1$ items in the test containing item 1 (this calculation also holds for $\mathbb{P}\big[M_2^j=0|\neg S_{12}^j\big]$). Likewise, $\mathbb{P}\big[M_1^j=0,M_2^j=0|\neg S_{12}^j\big]$ is calculated by counting the number of ways to choose the other $\frac{n\log2}{k}-2$ items in the tests (rows) of item 1 and item 2. Substituting \eqref{eq:prob_M1j_and_M_2j_both_1} into \eqref{eq:cov_M1j_and_M2j}, we have \eqref{eq:first_in_long_2}--\eqref{eq:cov_M1j_M2j} at the top of the next page, where (a) computes the probabilities in the same way as \eqref{eq:prob_M1j_and_M_2j_both_1}, and (b) follows by expanding the square and simplifying. 

The combinatorial terms in \eqref{eq:cov_M1j_M2j} can be bounded using manipulations that are elementary but tedious, so we state the resulting bound here and defer the proof to Appendix \ref{app:covariance}.

\begin{lemma} \label{lem:covariance}
    Under the preceding setup, we have $\Cov\big[M_1^j,M_2^j\big] \le O\big(\max\big\{\frac{1}{k},\frac{k}{n}\big\}\big)$.
\end{lemma}

Applying \eqref{eq:var_of_M_i^j} and Lemma \ref{lem:covariance} in \eqref{eq:var_M^j}, we obtain $\Var\big[M^j\big]=O\big(\max\big\{k,\frac{k^3}{n}\frac{}{}\big\}\big)$. By Chebyshev's inequality, it follows that
\begin{align}
    &\mathbb{P}\bigg[\Big|M^j-\E\big[M^j\big]\Big|\geq\frac{\E\big[M^j\big]}{\log n}\bigg]
    \leq\frac{\Var\big[M^j\big]}{\big(\E\big[M^j\big]/\log n\big)^2} \nonumber \\
    & \qquad =\frac{O\big(\max\big\{k,\frac{k^3}{n}\frac{}{}\big\}\big)\log^2n}{\Theta(k^2)}
    =\frac{\log^2n}{n^{\Omega(1)}},
\end{align}
where we used the fact that $k=\Theta\big(n^\theta\big)$ with $\theta\in(0,1)$. Hence, $M^j=\E[M^j]\big(1+O\big(\frac{1}{\log n}\big)\big)=\frac{k}{2}(1+o(1))$ for all $\mathsf{X}_j$ simultaneously with probability
\begin{align}
    \Big(1-\frac{\log^2n}{n^{\Omega(1)}}\Big)^{c\log n}
    \stackrel{(a)}{=}1-O\Big(\frac{c\log^3n}{n^{\Omega(1)}}\Big)
    =1-o(1),
\end{align}
where (a) uses the fact $(1+a)^b=(1+ab)(1+o(1))$ when $|a|<1$ and $ab=o(1)$.

\noindent\textbf{Step 3:} We condition on the events described in the general description of Step 3 following \eqref{eq:E[M_i^j]_general}.  Here the high-probability bounds dictate that $G=o(k)$ and $M^j=\frac{k}{2}(1+o(1))$.


\noindent\textbf{Step 4:} Substituting $r=c\log n$ into \eqref{eq:prob_def_i_masked_by_PD_general}, the conditional probability of defective item $i$ being masked by $\mathcal{PD}\setminus\{i\}$ in all sub-matrices is
    \begin{align}
        \prod_{j=1}^{r}\mathbb{P}_{\mathcal{E}}\big[\mathrm{M}_i^j\cup \mathrm{MG}_i^j\big]
        &\leq\bigg(\Big(\frac{M^j}{k}+\frac{G}{k}\Big)(1+o(1))\bigg)^{c\log n} \nonumber \\
        & \stackrel{(a)}{=}\Big(\frac{1}{2}(1+o(1))+\frac{o(k)}{k}\Big)^{c\log n} \nonumber \\
        &=\Big(\frac{1}{2}(1+o(1))\Big)^{c\log n} \nonumber \\
        &=\exp\big((-c\log2+o(1))\log n\big) \nonumber \\
        &=n^{-c\log2+o(1)},
    \end{align}
    where (a) substitutes $M^j=\frac{k}{2}(1+o(1))$ and $G=o(k)$.
    
    Taking the union bound over all defective items, the probability of any defective item being masked by $\mathcal{PD}\setminus\{i\}$ in all sub-matrices is at most $kn^{-c\log2+o(1)}=\Theta\big(n^{\theta-c\log2+o(1)}\big)$. The power of $n$ is below zero when $c\geq(1+\epsilon)\frac{\theta}{\log2}$ for some positive constant $\epsilon$. Together with our previous bound on $c$ (see \eqref{eq:G_prob_bound}), this requires us to choose $c\geq(1+\epsilon)\frac{\max\{\theta,1-\theta\}}{\log2}$, which means it suffices to have
\begin{align}
    T\geq(1+\epsilon)\frac{\max\{\theta,1-\theta\}}{\log^22}k\log n.
\end{align}
This completes the proof of Theorem \ref{thm:unconstrained_achievability}.

\subsection{Unconstrained Linear Regime with Approximate Recovery (Theorem \ref{thm:linear_achievability})} \label{sec:linear}

We again follow the four-step procedure to obtain the desired result. We assume that $s=\Theta(1)$ and $r=\Theta(1)$, but their values are otherwise generic.  Recall that each $\mathsf{X}_j$ is a $\frac{n}{s}\times n$ sub-matrix, and $\mathsf{X}$ is a $\frac{rn}{s}\times n$ matrix.  It suffices to prove the theorem for $s \ge 2$, since otherwise each $\mathsf{X}_j$ amounts to one-by-one testing and the FNR is trivially zero.

The following technical lemma will be used throughout the analysis.

\begin{lemma} \label{lem:simplification_lemma}
For any positive constants $a$ and $b$ (not depending on $n$) satisfying $a \le b$, we have
\begin{align}
    \prod_{i=a}^b\Big(1-\frac{k}{n-i}\Big)
    &=(1-p)^{b-a+1}\bigg(1-O\Big(\frac{1}{n}\Big)\bigg),
\end{align}
where $p = \frac{k}{n}$.
\end{lemma}

\begin{proof}
We have
\begin{align}
    \prod_{i=a}^b\Big(1-\frac{k}{n-i}\Big)
    &=\bigg(1-\frac{k}{n(1+O(\frac{1}{n}))}\bigg)^{b-a+1} \nonumber \\
    & \stackrel{(a)}{=}\bigg(1-p\bigg(1+O\Big(\frac{1}{n}\Big)\bigg)\bigg)^{b-a+1} \nonumber \\
    &=(1-p)\bigg(1-\frac{p}{1-p}\cdot O\Big(\frac{1}{n}\Big)\bigg)^{b-a+1} \nonumber \\
    & \stackrel{(b)}{=}(1-p)^{b-a+1}\bigg(1-O\Big(\frac{1}{n}\Big)\bigg),
\end{align}
where (a) substitutes $k=pn$ and uses $\big(1-O\big(\frac{1}{n}\big)\big)^{-1}=1+O\big(\frac{1}{n}\big)$, and (b) applies $\big(1-\frac{p}{1-p} O\big(\frac{1}{n}\big)\big)^{b-a+1}=1-O\big(\frac{1}{n}\big)$, noting that $\frac{p}{1-p}$ and $b-a+1$ are both constants.
\end{proof}

We now proceed with the main steps.

\noindent\textbf{Step 1:} Continuing from \eqref{eq:E[G]_general}, we have
\begin{align}
    \E[G]
    &=(n-k)\bigg(1-\prod_{i=1}^{s-1}\Big(1-\frac{k}{n-i}\Big)\bigg)^{r} \label{eq:start_binom_frac_simplification_linear} \\
    &\stackrel{(a)}{=}(n-k)\bigg(1-(1-p)^{s-1}\bigg(1-O\Big(\frac{1}{n}\Big)\bigg)\bigg)^r \\
    &\stackrel{(b)}{=}(n-k)\big(1-(1-p)^{s-1}\big)^r\bigg(1+O\Big(\frac{1}{n}\Big)\bigg), \label{eq:E[G]_linear}
\end{align}
where (a) applies Lemma \ref{lem:simplification_lemma}, and (b) applies $\big(1-O\big(\frac{1}{n}\big)\big)^{r}=\big(1-O\big(\frac{1}{n}\big)\big)$ for constant $r$. Next, we have
\begin{align}
    \text{Var}[G]
    &=\text{Var}\bigg[\sum_{i\in[n]\setminus\mathcal{K}}\mathrm{PD}_i\bigg] \nonumber \\
    &=(n-k)\text{Var}[\mathrm{PD}_1] \nonumber \\
    &\qquad + (n-k)(n-k-1)\text{Cov}[\mathrm{PD}_1, \mathrm{PD}_2], \label{eq:Var[G]_linear}
\end{align}
where here and in the rest of this step (step 1), we assume for notational convenience that items $1$ and $2$ are non-defective.\footnote{This should not be confused with the convention $\mathcal{K} = \{1,\dotsc,k\}$ and $\mathcal{G} = \{k+1,\dotsc,k+G\}$ used when analyzing the defective items in other steps.  The precise indices are inconsequential and merely a matter of notational convenience, so there is no contradiction in using indices $1$ and $2$ for non-defectives here but for defectives in other parts.}  We proceed to evaluate the variance and covariance terms separately. We have
\begin{align}
    &\text{Var}[\mathrm{PD}_1] \nonumber \\
    &=\mathbb{E}\big[\mathrm{PD}_1^2\big]-\mathbb{E}\big[\mathrm{PD}_1\big]^2 \nonumber \\
    &=\mathbb{P}[\mathrm{PD}_1=1]-\mathbb{P}[\mathrm{PD}_1=1]^2 \nonumber \\
    &\stackrel{(a)}{=}\big((1-(1-p)^{s-1})^r-(1-(1-p)^{s-1})^{2r}\big)(1+o(1)) \nonumber \\
    &=\Theta(1), \label{eq:Var[PD_1]_linear}
\end{align}
where (a) uses the same steps as \eqref{eq:prob_nondef_in_neg_test_general}--\eqref{eq:E[PD_i]_general} followed by Lemma \ref{lem:simplification_lemma}. Next, we have
\begin{align}
    &\text{Cov}[\mathrm{PD}_1, \mathrm{PD}_2] \nonumber \\
    &=\mathbb{E}[\mathrm{PD}_1 \mathrm{PD}_2]-\mathbb{E}[\mathrm{PD}_1]\mathbb{E}[\mathrm{PD}_2] \\
    &=\mathbb{P}[\mathrm{PD}_1=1,\mathrm{PD}_2=1] \nonumber \\
        &\qquad -(1-(1-p)^{s-1})^{2r}\bigg(1+O\Big(\frac{1}{n}\Big)\bigg). \label{eq:cov_init}
\end{align}
Focusing on the first term, we note that the event $(\mathrm{PD}_1=1)\cap (\mathrm{PD}_2=1)$ holds if items $1$ and $2$ are both contained in a positive test in each sub-matrix. Since the sub-matrices are sampled independently, we can consider them separately. Let $\mathrm{PD}_1^j, \mathrm{PD}_2^j$ respectively be the events that items $1,2$ are contained in a positive test in sub-matrix $\mathsf{X}_j$. Then $\mathbb{P}[\mathrm{PD}_1=1,\mathrm{PD}_2=1]=\mathbb{P}\big[\mathrm{PD}_1^1=1,\mathrm{PD}_2^1=1\big]^r$.

Let $S_{12}^j$ be the event that items $1,2$ are placed into the same test in $\mathsf{X}_j$. Similar to before (see \eqref{eq:S12j_uncon_sublinear}), we have $\mathbb{P}\big[S_{12}^j\big] = \frac{s-1}{n-1}$. Conditioning on event $S_{12}^j$, we have:
\begin{align}
    &\mathbb{P}\big[\mathrm{PD}_1^1=1,\mathrm{PD}_2^1=1|S_{12}^j\big] 
    =1-\frac{{n-k-2 \choose s-2}}{{n-2 \choose s-2}} \nonumber \\
    &\qquad =\big(1-(1-p)^{s-2}\big)\bigg(1+O\Big(\frac{1}{n}\Big)\bigg), \label{eq:conditional_PD_1_and_PD_2_linear1}
\end{align}
\begin{figure*}[t]
    \normalsize
    \setcounter{mytempeqncnt}{\value{equation}}
    \begin{align}
        \mathbb{P}\big[\mathrm{PD}_1^1=1,\mathrm{PD}_2^1=1|\neg S_{12}^j\big]
        &=1-\mathbb{P}\big[\mathrm{PD}_1^1=0\cup \mathrm{PD}_2^1=0|\neg S_{12}^j\big] \label{eq:long_start3} \\
        &\stackrel{(a)}{=}1-\Big(\mathbb{P}\big[\mathrm{PD}_1^1=0|\neg S_{12}^j\big]
        +\mathbb{P}\big[\mathrm{PD}_2^1=0|\neg S_{12}^j\big] \nonumber \\
        &\qquad-\mathbb{P}\big[\mathrm{PD}_1^1=0,\mathrm{PD}_2^1=0|\neg S_{12}^j\big]\Big) \\
        &= 1 - \left(2\frac{{n-k-2 \choose s-1}}{{n-2 \choose s-1}} - \frac{{n-k-2 \choose s-1}}{{n-2 \choose s-1}} \cdot \frac{{n-k-s-1\choose s-1}}{{n-s-1\choose s-1}}\right)\\
        &\stackrel{(b)}{=}1-\Big(2(1-p)^{s-1} - (1-p)^{2(s-1)}\Big)\bigg(1-O\Big(\frac{1}{n}\Big)\bigg) \\
        &\stackrel{(b)}{=}\big(1-(1-p)^{s-1}\big)^2\bigg(1+O\Big(\frac{1}{n}\Big)\bigg), \label{eq:conditional_PD_1_and_PD_2_linear2}
    \end{align}
    \hrulefill 
    \begin{align}
        \mathbb{P}\big[\mathrm{PD}_1^1=1,\mathrm{PD}_2^1=1\big]
        &=\mathbb{P}\big[S_{12}^j\big]\mathbb{P}\big[\mathrm{PD}_1^1=1,\mathrm{PD}_2^1=1|S_{12}^j\big]
        +\mathbb{P}[\neg S_{12}^j]\mathbb{P}\big[\mathrm{PD}_1^1=1,\mathrm{PD}_2^1=1|\neg S_{12}^j\big] \label{eq:long_start4} \\
        &\stackrel{(a)}{=}\bigg(\frac{s-1}{n-1}\cdot(1-(1-p)^{s-2})+\Big(1-\frac{s-1}{n-1}\Big)\cdot(1-(1-p)^{s-1})^2\bigg)\bigg(1+O\Big(\frac{1}{n}\Big)\bigg) \\
        &\stackrel{(b)}{=}\bigg(O\Big(\frac{1}{n}\Big)+\Big(1-O\Big(\frac{1}{n}\Big)\Big)\big(1-(1-p)^{s-1}\big)^2\bigg)\bigg(1+O\Big(\frac{1}{n}\Big)\bigg) \\
        &\stackrel{(c)}{=}\big(1-(1-p)^{s-1}\big)^2\bigg(1+O\Big(\frac{1}{n}\Big)\bigg), \label{eq:joint_p}
    \end{align}
    \hrulefill
    \vspace*{4pt}
\end{figure*}
where the final equality uses the same steps as \eqref{eq:prob_nondef_in_neg_test_general} followed by Lemma \ref{lem:simplification_lemma}. On the other hand, under the complement event, we have \eqref{eq:long_start3}--\eqref{eq:conditional_PD_1_and_PD_2_linear2} at the top of the next page, where (a) uses the the inclusion-exclusion principle, (b) uses the same steps as \eqref{eq:prob_nondef_in_neg_test_general} followed by Lemma \ref{lem:simplification_lemma}, and (c) uses the fact that $s$ and $p$ are constant. Hence, by the law of total probability, we have \eqref{eq:long_start4}--\eqref{eq:joint_p} at the top of the next page,
where (a) substitutes $\mathbb{P}\big[S_{12}^j\big] = \frac{s-1}{n-1}$ along with \eqref{eq:conditional_PD_1_and_PD_2_linear1} and \eqref{eq:conditional_PD_1_and_PD_2_linear2}, (b) applies $\frac{s-1}{n-1}=O\big(\frac{1}{n}\big)$, and both (b) and (c) use the fact that $s$ and $p$ are constant. Combining \eqref{eq:joint_p} with \eqref{eq:cov_init} gives
\begin{align}
    \text{Cov}[\mathrm{PD}_1,\mathrm{PD}_2]
    &=O\Big(\frac{1}{n}\Big), \label{eq:Cov[PD_1,PD_2]_linear}
\end{align}
since the leading terms cancel and only leave the $O\big(\frac{1}{n}\big)$ remainder.
Substituting \eqref{eq:Var[PD_1]_linear} and \eqref{eq:Cov[PD_1,PD_2]_linear} into \eqref{eq:Var[G]_linear}, we obtain
\begin{align}
    \text{Var}[G]
    &\le\Theta(n) \Theta(1) + \Theta\big(n^2\big)O\big(n^{-1}\big)
    =\Theta(n).
\end{align}
Since $\E[G] = \Theta(n)$ (see \eqref{eq:E[G]_linear}), it follows from Chebyshev's Inequality that
\begin{align}
    \mathbb{P}\bigg[\big|G-\mathbb{E}[G]\big|\ge \frac{\E[G]}{\log n}\bigg] 
    &\le\frac{\Var[G]\log^2n}{(\E[G])^2}
    =\frac{\log^2n}{n},
\end{align}
implying that $G =(n-k)\big(1-(1-p)^{s-1}\big)^r(1+o(1))$ with probability $1-o(1)$.

\noindent\textbf{Step 2:} Continuing from \eqref{eq:E[M_i^j]_general}, we have
\begin{align}
    \E\big[M^j\big]
    &=k\bigg(1-\prod_{i=1}^{s-1}\Big(1-\frac{k-1}{n-i}\Big)\bigg) \nonumber \\
    &=k\big(1-(1-p)^{s-1}\big)(1+o(1)),
\end{align}
where the last equality uses the same steps as those in \eqref{eq:start_binom_frac_simplification_linear}--\eqref{eq:E[G]_linear}. We can use a near-identical approach as Step 1 to establish that $\Var[M^j] = \Theta(n)$, and again applying Chebyshev's inequality, we have $M^j = k(1-(1-p)^{s-1})(1+o(1))$ for all $\mathsf{X}_j$ simultaneously with probability $(1-o(1))^r=1-o(1)$ (since $r=\Theta(1)$).

\noindent\textbf{Step 3:} We again condition on the events described in the general description of Step 3 following \eqref{eq:E[M_i^j]_general}.  Here the high-probability bounds dictate that $G=k(1-(1-p)^{s-1})^r(1+o(1))$ and $M^j = k(1-(1-p)^{s-1})(1+o(1))$.


\noindent\textbf{Step 4:} Continuing from \eqref{eq:prob_def_i_masked_by_PD_general}, the conditional probability of defective item $i$ being masked by $\mathcal{PD}\setminus\{i\}$ in all sub-matrices (which is precisely the FNR) satisfies
\begin{align}
    &\prod_{j=1}^{r}\mathbb{P}_{\mathcal{E}}\big[\mathrm{M}_i^j\cup \mathrm{MG}_i^j\big] \nonumber \\
    &\leq\bigg(\Big(\frac{M^j}{k}+\frac{G}{k}\Big)(1+o(1))\bigg)^r \\
    &\stackrel{(a)}{=}\bigg(\big(1-(1-p)^{s-1}\big)+\frac{(n-k)(1-(1-p)^{s-1})^r}{k}\bigg)^r \nonumber \\
        &\hspace*{5.5cm}\times(1+o(1)) \\
    &\stackrel{(b)}{=}\bigg(\big(1-(1-p)^{s-1}\big)+\frac{(1-p)(1-(1-p)^{s-1})^r}{p}\bigg)^r \nonumber \\
    &\hspace*{5.5cm}\times(1+o(1)),
\end{align}
where (a) substitutes $M^j = k(1-(1-p)^{s-1})(1+o(1))$ and $G=(n-k)\big(1-(1-p)^{s-1}\big)^r(1+o(1))$, and (b) substitutes $k=pn$. This completes the proof of Theorem \ref{thm:linear_achievability}.

\subsection{Size-Constrained Sub-Linear Regime with Exact Recovery (Theorem \ref{thm:sparse_unconstrained_achievability})} \label{sec:con_sublinear}

In this setting, to reduce the number of constants throughout, we consider $k=n^{\theta}$ and $\rho = \big(\frac{n}{k}\big)^\beta$ (i.e., implied constants of one in their scaling laws), but the general case follows with only minor changes. We again follow the above four-step procedure to obtain our result, and begin by selecting the relevant parameters. We choose $s=\rho$, and let $r$ be a constant (not depending on $n$) to be chosen later. This results in each $\mathsf{X}_j$ having size $\frac{n}{\rho}\times n$, and $\mathsf{X}$ having size $\frac{rn}{\rho}\times n$.

\noindent\textbf{Step 1:} Substituting $s=\rho$ into \eqref{eq:E[G]_general}, we obtain
\begin{align}
    \E[G]
    &=\Theta\bigg(n\bigg(1-\prod_{i=1}^{\rho-1}\Big(1-\frac{k}{n-i}\Big)\bigg)^{r}\bigg) \label{eq:start_binom_frac_simplification_con+sublinear} \\
    &\stackrel{(a)}{=}\Theta\bigg(n\bigg(1-\Big(1-\frac{k}{n(1-o(1))}\Big)^{\rho-1}\bigg)^{r}\bigg) \\
    &\stackrel{(b)}{=}\Theta\bigg(n\bigg(\frac{k\rho}{n}\bigg)^{r}\bigg), \label{eq:E[G]_con_sublinear}
\end{align}
where (a) uses $i\leq\rho-1=o(n)$, and (b) uses the fact $(1+a)^b=(1+ab)(1+o(1))$ when $|a|<1$ and $ab=o(1)$ (recall that $\rho \to \infty$ with $\rho=o\big(\frac{n}{k}\big)$).  We now consider two cases:
\begin{itemize}
    \item For $\theta\geq\frac{1}{2}$, we use Markov's inequality to obtain
    \begin{align}
        \mathbb{P}\Big[G\geq\frac{k^2\rho}{n}\log n\Big]
        &\leq\frac{\E[G]}{\frac{k^2\rho}{n}\log n}
        \leq\frac{n\big(\frac{k\rho}{n}\big)^r}{\frac{k^2\rho}{n}\log n} \nonumber \\
        &\stackrel{(a)}{=}\frac{n^{(1-\theta)(2-\beta)-r(1-\theta)(1-\beta)}}{\log n}, \label{eq:P[G_bound_theta>=1/2]}
    \end{align}
    where (a) substitutes $\rho=(n/k)^\beta$ and $k=n^\theta$. Note that when $r\geq\frac{2-\beta}{1-\beta}$, the expression in \eqref{eq:P[G_bound_theta>=1/2]} scales as $O\big(\frac{1}{\log n}\big)$. 
    \item For $\theta<\frac{1}{2}$, we similarly have
    \begin{align}
        \mathbb{P}[G\geq k^\beta\log n]
        &\leq\frac{\E[G]}{k^\beta\log n}
        \leq\frac{n\big(\frac{k\rho}{n}\big)^r}{k^\beta\log n} \nonumber \\
        &\stackrel{(a)}{=}\frac{n^{1-r(1-\theta)(1-\beta)-\theta\beta}}{\log n}, \label{eq:P[G_bound_theta<1/2]}
    \end{align}
    where (a) substitutes $\rho=(n/k)^\beta$ and $k=n^\theta$. When $r\geq\frac{1-\theta\beta}{(1-\theta)(1-\beta)}$, the expression in \eqref{eq:P[G_bound_theta<1/2]} scales as $O\big(\frac{1}{\log n}\big)$. As a side-note, to see why we split $\theta$ into two cases here, we note that $\frac{2-\beta}{1-\beta}\leq\frac{1-\theta\beta}{(1-\theta)(1-\beta)}$ if and only if $\theta\geq\frac{1}{2}$ (for any $\beta\in[0,1)$).
\end{itemize}

\noindent\textbf{Step 2:} Substituting $s=\rho$ into \eqref{eq:E[M_i^j]_general}, we have
\begin{align}
    \E\big[M^j\big]
    &=k\bigg(1-\prod_{i=1}^{\rho-1}\Big(1-\frac{k-1}{n-i}\Big)\bigg)
    =\frac{k^2\rho}{n}(1+o(1)),
\end{align}
where the last equality uses the same steps as those in \eqref{eq:start_binom_frac_simplification_con+sublinear}--\eqref{eq:E[G]_con_sublinear}. By Markov's inequality, it follows that
\begin{align}
    \mathbb{P}\Big[M^j\geq\frac{k^2\rho}{n}\log n\Big]
    \leq\frac{\E[M^j]}{\frac{k^2\rho}{n}\log n}
    \leq\frac{1+o(1)}{\log n}. \label{eq:D's_bound}
\end{align}
Turning to the desired event for all $r$ sub-matrices simultaneously, we have $M^j<\frac{k^2\rho}{n}\log n$ with probability at least $\big(1-\frac{1+o(1)}{\log n}\big)^r=1-O\big(\frac{1}{\log n}\big)$, using the fact that $r=\Theta(1)$.

\noindent\textbf{Step 3:} We again condition on the events described in the general description of Step 3 following \eqref{eq:E[M_i^j]_general}.  Here the high-probability bounds dictate that
\begin{align}
    G&\le
    \begin{cases}
    \frac{k^2\rho}{n}\log n, \text{ if $\theta\geq\frac{1}{2}$} \\
    k^\beta\log n, \text{ if $\theta<\frac{1}{2}$},
    \end{cases} \label{eq:g_value_con_sublinear}
\end{align}
and $M^j \le \frac{k^2\rho}{n}\log n$ for all $\mathsf{X}_j$ simultaneously.  

\noindent\textbf{Step 4:} Continuing from \eqref{eq:prob_def_i_masked_by_PD_general}, the probability of defective item $i$ being masked by $\mathcal{PD}\setminus\{i\}$ in all sub-matrices is
\begin{align}
    &\prod_{j=1}^{r}\mathbb{P}_{\mathcal{E}}\big[\mathrm{M}_i^j\cup \mathrm{MG}_i^j\big] \nonumber \\
    &\leq\bigg(\Big(\frac{M^j}{k}+\frac{G}{k}\Big)(1+o(1))\bigg)^r \\
    &\le
    \begin{cases}
    \Big(O\Big(\frac{k\rho}{n}\log n\Big)\Big)^r &\text{if $\theta\geq\frac{1}{2}$} \\
    \Big(O\Big(\frac{k\rho}{n}\log n+k^{-(1-\beta)}\log n\Big)\Big)^r &\text{if $\theta<\frac{1}{2}$},
    \end{cases}
\end{align}
where we substituted $M^j \le \frac{k^2\rho}{n}\log n$ and $G$ as shown in \eqref{eq:g_value_con_sublinear} and applies some simplifications.

Taking the union bound over all $k$ defective items and equating the resulting bound with a target value of $\frac{1}{\log n}$, we obtain the following conditions for the desired events to hold with probability at least $1 - \frac{1}{\log n}$:
\begin{align}
    k\bigg(O\Big(\frac{k\rho}{n}\log n\Big)\bigg)^r\leq\frac{1}{\log n} &\text{ if $\theta\geq\frac{1}{2}$} \\
    k\bigg(O\Big(\frac{k\rho}{n}\log n+k^{-(1-\beta)}\log n\Big)\bigg)^r\leq\frac{1}{\log n} &\text{ if $\theta<\frac{1}{2}$}.
\end{align}
In other words, there exist constants $C_1$ and $C_2$ such that it suffices to have
\begin{align}
    k\bigg(C_1\frac{k\rho}{n}\log n\bigg)^r\leq\frac{1}{\log n} &\text{ if $\theta\geq\frac{1}{2}$} \nonumber \\
    k\bigg(C_2\max\Big\{\frac{k\rho}{n}\log n,k^{-(1-\beta)}\log n\Big\}\bigg)^r\leq\frac{1}{\log n} &\text{ if $\theta<\frac{1}{2}$}.
\end{align}
Substituting $\rho=\big(\frac{n}{k}\big)^{\beta}$ and $k=n^\theta$, and performing some simplifications, we obtain the following conditions:
\begin{itemize}
    \item For $\theta\geq\frac{1}{2}$:
        \begin{align}
            n^\theta\Big(C_1n^{-(1-\theta)(1-\beta)+o(1)}\Big)^r\leq\frac{1}{\log n}.  \label{eq:eps_n_eq1}
        \end{align}
    \item For $\theta<\frac{1}{2}$:
    \begin{multline}
        n^\theta\Big(C_2\max\big\{n^{-(1-\theta)(1-\beta)+o(1)},n^{-\theta(1-\beta)+o(1)}\big\}\Big)^r \\ \leq\frac{1}{\log n}.  \label{eq:eps_n_eq2}
    \end{multline}
\end{itemize}
Observe that for any integer $r$ satisfying
\begin{align}
    r&>\frac{\theta}{(1-\theta)(1-\beta)},
\end{align}
the exponent of $n$ on the left-hand side of \eqref{eq:eps_n_eq1} becomes negative, so that \eqref{eq:eps_n_eq1} is satisfied. Similarly, observe that for any integer $r$ satisfying
\begin{align}
    r&>\max\bigg\{\frac{\theta}{(1-\theta)(1-\beta)},\frac{1}{1-\beta}\bigg\},
\end{align}
the exponent of $n$ on the left-hand side of \eqref{eq:eps_n_eq2} becomes negative, so that \eqref{eq:eps_n_eq2} is indeed satisfied. Hence, together with the bounds on $r$ from step 1, we require $r$ to be an integer such that the following conditions hold:
\begin{itemize}
    \item For $\theta \ge \frac{1}{2}$:
    \begin{align}
        r>\frac{\theta}{(1-\theta)(1-\beta)}, ~ r\geq\frac{2-\beta}{1-\beta}.
        \label{eq:r_condition_theta>=1/2}
    \end{align}
    \item For $\theta < \frac{1}{2}$:
    \begin{align}
        r>\max\Big\{\frac{\theta}{(1-\theta)(1-\beta)},\frac{1}{1-\beta}\Big\}, r\geq\frac{1-\theta\beta}{(1-\theta)(1-\beta)}. \label{eq:r_condition_theta<1/2}
    \end{align}
\end{itemize}
Note that \eqref{eq:r_condition_theta<1/2} simplifies to $r\geq\frac{1-\theta\beta}{(1-\theta)(1-\beta)}$ alone, because for any $\theta\in\big[0,\frac{1}{2}\big)$ and $\beta\in[0,1)$, we have (i) $1-\theta\beta>\theta$, which implies that $\frac{1-\theta\beta}{(1-\theta)(1-\beta)}>\frac{\theta}{(1-\theta)(1-\beta)}$, and (ii) $\frac{1-\theta\beta}{1-\theta}>1$, which implies that $\frac{1-\theta\beta}{(1-\theta)(1-\beta)}>\frac{1}{1-\beta}$.

Finally, we used a total of $rn/\rho$ tests and incurred a total error probability of at most $O\big(\frac{1}{\log n}\big)$. This completes the proof of Theorem \ref{thm:sparse_unconstrained_achievability}.

\section{Converse Analysis for the Size-Constrained Setting} \label{sec:uncon_sublinear_converse}

In this section, we prove Theorem \ref{thm:rho_converse}.  We prove the results corresponding to $r=\big\lceil\frac{1-(1-\theta)(2\beta+1)}{(1-\theta)(1-\beta)}\big\rceil$ and $r=2$ separately in Section \ref{sec:c=ceiling_term} and Section \ref{sec:c=2} respectively; the remaining term $\frac{1}{1-\beta}$ is already known from \cite{Ven19}.

\begin{remark} \label{rem:conv}
    Compared to the achievability part, the analysis in this section builds more closely on that of the prior work \cite{Oli20} handling the regime $\rho = O(1)$.  In particular, for the first part (Section \ref{sec:c=ceiling_term}), we mostly follow \cite{Oli20} but with different choices of parameters to carefully account for the scaling of $\rho$. On the other hand, for the second part (Section \ref{sec:c=2}), more substantial changes are needed:  In \cite{Oli20}, the scaling $\rho = O(1)$ makes it relatively easier to identify positive tests with multiple items of degree one (leading to failure), whereas in our setting with $\rho = \omega(1)$ the argument is somewhat more lengthy and technical, though still adopts a similar high-level approach.
\end{remark}

\subsection{Part I} \label{sec:c=ceiling_term}

To reduce the number of constants throughout, we consider $k=n^\theta$ and test sizes $\rho=\big(\frac{n}{k}\big)^\beta=n^{\beta(1-\theta)}$ (i.e.,
implied constants of one in their scaling laws), but the general case follows with only minor changes. Without loss of generality, we may assume that $\theta>\frac{1+\beta}{2+\beta}$; this is because if $\theta\leq\frac{1+\beta}{2+\beta}$, then $\big\lceil\frac{1-(1-\theta)(2\beta+1)}{(1-\theta)(1-\beta)}\big\rceil\leq1$, which implies that it will not be the maximum among the three constant terms in \eqref{eq:two_cases_previously}.  We make use of the notion of masked items in Definition \ref{def:masked_def_item}, and we are now interested in characterizing both masked defectives and masked non-defectives.


Let the true defectivity vector be $\mathbf{u}\in\{0,1\}^n$, where $u_i=1$ indicates that the $i$-th item is defective, and $u_i=0$ otherwise. Furthermore, let the test outcomes be represented by the vector $\mathbf{y}\in\{0,1\}^T$, where $y_i=1$ denotes that the $i$-th test is positive and $y_i=0$ otherwise. We also define
\begin{align}
    d^-&=\bigg\lceil\frac{1-(1-\theta)(2\beta+1)}{(1-\theta)(1-\beta)}\bigg\rceil-1,
    \text{ and } \nonumber \\
    d^+&=\bigg\lceil\frac{1-(1-\theta)(2\beta+1)}{(1-\theta)(1-\beta)}\bigg\rceil, \label{eq:d_def}
\end{align}
and we note that our assumption $\theta>\frac{1+\beta}{2+\beta}$ implies that $d^->0$.

We can visualize any non-adaptive group testing design matrix $\mathsf{X}$ as a bipartite graph $\mathscr{G}=(V\cup F,E)$ with $|F|=T$ factor nodes $\{a_1,\dots,a_T\}$ (tests) and $|V|=n$ variable nodes $\{u_1,\dots,u_n\}$ (items). An edge between an item $u_i$ and test $a_j$ indicates that $u_i$ takes part in $a_j$. We let $\{\partial_\mathscr{G}a_1,\dots,\partial_\mathscr{G}a_T\}$ and $\{\partial_\mathscr{G}u_1,\dots,\partial_\mathscr{G}u_n\}$ denote the neighborhoods in $\mathscr{G}$; 
the sparsity constraint on the test implies $|\partial_\mathscr{G}a_j|\leq\rho$. Finally, let $V_{1+}(\mathscr{G})$ be the set of defectives that are masked, $V_{0+}(\mathscr{G})$ be the set of non-defectives that are masked, and $V_{+}(\mathscr{G})=V_{1+}(\mathscr{G})\cup V_{0+}(\mathscr{G})$. We have the following auxiliary results.

\begin{lemma} \label{lem:Pe_from_V+} \textup{\cite[Claim 2.3]{Oli20}}
Conditioned on any non-empty realizations of the sets $V_{1+}(\mathscr{G})$ and $V_{0+}(\mathscr{G})$, any decoding procedure fails with probability at least $1-\frac{1}{|V_{1+}(\mathscr{G})|\cdot|V_{0+}(\mathscr{G})|}$.
\end{lemma}

\begin{lemma} \label{lem:mul_chernoff_bound} \textup{(Multiplicative Chernoff bound \cite[Thm.~4.2]{Mot10})}
Suppose $X_1,\dots,X_m$ are independent random variables taking values in $\{0,1\}$. Let $X=\sum_{i=1}^mX_i$. Then for any $\epsilon\in(0,1)$,
\begin{align}
    \mathbb{P}[X\leq(1-\epsilon)\E[X]]
    &\leq\exp\bigg(-\frac{\epsilon^2\E[X]}{2}\bigg).
\end{align}
\end{lemma}

Our goal is to show that with $T=(1-\epsilon)\frac{d^+n}{\rho}$ for some constant $\epsilon>0$, we have $|V_{1+}(G)|,|V_{0+}(G)|=\omega(1)$ with high probability (i.e., with probability $1-o(1)$), which implies (using Lemma \ref{lem:Pe_from_V+}) that $P_e=1-o(1)$.  As a stepping stone to our result for the combinatorial prior, we first study the i.i.d.~prior, whose defectivity vector we denote by $\mathbf{u}^*$, such that each entry is one independently with probability $p=\frac{k-\sqrt{k}\log n}{n}$. The following existing result allows us to transfer from the latter prior to the former.

\begin{lemma} \label{lem:comb_vs_iid_prior} \textup{\cite[Corollary 3.6]{Oli20}}
Given non-negative integers $C_1,C_2$ and fixed $\epsilon_1,\epsilon_2\in(0,1)$, if the modified model (i.i.d.~prior) satisfies
\begin{align}
   & \mathbb{P}\big[|V_{1+}(\mathscr{G},\mathbf{u}^*)|>2C_1\big]\geq1-\epsilon_1
    \text{ and } \nonumber \\
    &\mathbb{P}\big[|V_{0+}(\mathscr{G},\mathbf{u}^*)|>2C_2\big]\geq1-\epsilon_2, \label{eq:combi_to_iid_def}
\end{align}
then the original model (combinatorial prior) satisfies
\begin{align}
    &\mathbb{P}\big[|V_{1+}(\mathscr{G},\mathbf{u})|>C_1\big]\geq1-\epsilon_1-o(1)
    \text{ and } \nonumber \\
    &\mathbb{P}\big[|V_{0+}(\mathscr{G},\mathbf{u})|>C_2\big]\geq1-\epsilon_2-o(1). \label{eq:combi_to_iid_nondef}
\end{align}
\end{lemma}

In view of this result, we proceed by working with $\mathbf{u}^*$ instead of $\mathbf{u}$.  We will also use the following key fact \cite{Coj19}: Whenever items have distance at least 6 in the underlying graph, the events of being masked are independent. Leveraging on this fact, we introduce a procedure for obtaining a set $\mathcal{B}$ of size $N$ (to be specified shortly), such that each item has a degree of at most $d^-$, and the pairwise distance between items is at least 6. The procedure is as follows:
\begin{enumerate}
    \item Create $\mathscr{G}_0$ from $\mathscr{G}$ by executing the following: 
    \begin{enumerate}
        \item Remove all items whose degree in $\mathscr{G}$ is greater than $\log n$.
        \item Identify all tests whose degree in $\mathscr{G}$ is less than $\frac{\rho}{\log n}$, and simultaneously remove all of those tests and their respective items.
    \end{enumerate}
    \item For $i=1,\dots,N$, where $N=\frac{n^{1-2\beta(1-\theta)}}{\log^3 n}$:
    \begin{enumerate}
        \item Arbitrarily select an item in $\mathscr{G}_{i-1}$ whose degree in $\mathscr{G}$ is at most $d^-$ (assuming one exists).
        \item Extract the item and add it to the set $\mathcal{B}$.
        \item Remove the extracted item's tests in the first and third neighborhood, its items in the second and fourth neighborhood, and the extracted item itself.  Let $\mathscr{G}_i$ be the resulting graph.
    \end{enumerate}
\end{enumerate}
We now proceed to analyze each of the steps.

\textbf{Analysis of Step 1:} 
In Step 1(a), we remove all items with degree greater than $\log n$.  Then, in Step 1(b), we further remove tests with degree less than $\frac{\rho}{\log n}$ and their items.  We are left with a sub-graph $\mathscr{G}_0$ of the original graph, whose nodes' degree properties turn out to be convenient for the analysis.

Before proceeding with $\mathscr{G}_0$, we need to understand its number of items.
%
To do so, we investigate how many items of various kinds could have been removed from $\mathscr{G}$.  We first note that the number of edges in $\mathscr{G}$ contributed by items with degree greater than $\log n$ is upper bounded by the total number of edges, which equals
\begin{align}
    \sum_{u\in V(\mathscr{G})}|\partial_{\mathscr{G}} u|
    =\sum_{a\in F(\mathscr{G})}|\partial_{\mathscr{G}} a|
    \leq T\rho
    =(1-\epsilon)d^+n.
\end{align}
Hence, the number of items in $\mathscr{G}$ with degree greater than $\log n$ (and hence removed in Step 1(a)) is at most $\frac{(1-\epsilon)d^+n}{\log n}=o(n)$.

Next, we note that the number of tests with degree less than $\frac{\rho}{\log n}$ removed in Step 1(b) is trivially no higher than the total number of tests $\frac{(1-\epsilon)d^+n}{\rho}$. Hence, the number of edges contributed by those tests is at most $\frac{\rho}{\log n}\cdot \frac{(1-\epsilon)d^+n}{\rho}=o(n)$. This implies that the number of items that are removed in Step 1(b) is at most $o(n)$.  
Combining these observations, we conclude that $\mathscr{G}_0$ has $n(1-o(1))$ items, all with degree at most $\log n$ in $\mathscr{G}_0$.  Moreover, every test in $\mathscr{G}_0$ has degree between $\frac{\rho}{\log n}$ and $\rho$ in $\mathscr{G}$ (though their degree in $\mathscr{G}_0$ itself could be below $\frac{\rho}{\log n}$).

\begin{lemma} \label{lem:items_with_at_most_d^-_degree}
Under the preceding setup with $T \le (1-\epsilon)\frac{d^+n}{\rho}$, the number of items in $\mathscr{G}_0$ with degree at most $d^-$ in $\mathscr{G}$ is $\Theta(n)$.
\end{lemma}

\begin{proof}
Suppose for contradiction that the number of items in $\mathscr{G}_0$ with degree at most $d^-$ in $\mathscr{G}$ is $o(n)$.  Then, the number of items in $\mathscr{G}_0$ with degree at least $d^+$ in $\mathscr{G}$ is $n(1-o(1))$. It follows that the number of edges (in $\mathscr{G}$) contributed by these items is at least $d^+n(1-o(1))$, which implies that the number of tests required for those edges is at least
\begin{align}
    \frac{d^+n(1-o(1))}{\rho}
    > (1-\epsilon)\frac{d^+n}{\rho},
\end{align}
for sufficiently large $n$, contradicting the assumption on $T$.
\end{proof}

\textbf{Analysis of Step 2:} Due to the degree upper bounds (in $\mathscr{G}_0$) established above, in each iteration, we remove at most
\begin{align}
    \rho\log n+\rho^2\log^2n=\Theta\big(n^{2\beta(1-\theta)}\log^2n\big)
\end{align}
items.  This scales as $o(n)$ when $\theta>1-\frac{1}{2\beta}$ with $\beta\in(0,1)$, which is satisfied due to our assumption $\theta>\frac{1+\beta}{2+\beta}=1-\frac{1}{2+\beta}$ (note that $2+\beta\geq2\beta$). In total, after $N=\frac{n^{1-2\beta(1-\theta)}}{\log^3 n}$ iterations, we removed a number of items scaling as
\begin{align}
    N\cdot\Theta\big(n^{2\beta(1-\theta)}\log^2n\big)=\Theta\bigg(\frac{n}{\log n}\bigg)=o(n). \label{eq:d^-_deg_items}
\end{align}
Since we have $\Theta(n)$ items with degree at most $d^-$ in $\mathscr{G}_0$ (see Lemma \ref{lem:items_with_at_most_d^-_degree}), and we remove only $o(n)$ items (from \eqref{eq:d^-_deg_items}), it follows that we will always be able to extract $N$ items, i.e., the item described in Step 2(a) is always guaranteed to exist under our assumptions.

{\bf Analysis of the masking probability.} We now focus our attention on the set $\mathcal{B}$. For each item $u\in\mathcal{B}$, we have 
\begin{align}
    \mathbb{P}[u\in V_{+}(\mathscr{G},\mathbf{u}^*)]
    &\stackrel{(a)}{\geq}\prod_{a\in\partial_{\mathscr{G}} u}\Big(1-(1-p)^{|\partial_{\mathscr{G}} a|-1}\Big) \label{eq:mask_step_a}\\
    &\stackrel{(b)}{=}\prod_{a\in\partial_{\mathscr{G}_0} u}\Big(1-(1-p)^{|\partial_{\mathscr{G}} a|-1}\Big) \label{eq:mask_step_b}\\
    &\stackrel{(c)}{=} \prod_{a\in\partial_{\mathscr{G}_0} u}\Big( p(|\partial_{\mathscr{G}} a|-1)(1+o(1)) \Big) \\
    &\stackrel{(d)}{\geq}\bigg(\frac{p\rho}{\log n}(1+o(1))\bigg)^{|\partial_{\mathscr{G}_0} u|} \\
    &\stackrel{(e)}{\geq}\bigg(\frac{p\rho}{\log n}(1+o(1))\bigg)^{d^-},
\end{align}
where:
\begin{itemize}
    \item (a) follows from the fact that the events of $u$ being masked in each of its tests satisfy an increasing property\footnote{Namely, marking any additional item(s) as defective can only increase (or keep unchanged) the probability of the masking event associated with each test $a \in \partial_{\mathscr{G}} u$.} \cite[Lemma 4]{Ald18}; this is a sufficient condition for applying the FKG inequality \cite{For71}, which lower bounds the probability by the expression that would hold under independence across $a \in \partial_{\mathscr{G}} u$.  Note also that $|\partial_{\mathscr{G}} a| \geq \frac{\rho}{\log n} \ge 2$ because our procedure (Step 1(b)) ensures that item $u$ is not an item that is individually tested in $\mathscr{G}$.
    \item (b) applies $\partial_{\mathscr{G}}u = \partial_{\mathscr{G}_0}u$, which holds because transforming $\mathscr{G}$ into $\mathscr{G}_0$ causes $u$ to stay in the same tests. This is because whenever we remove a test in Step 1 of the procedure, we also remove all the items in the test.
    \item (c) follows by first noting that $p|\partial_{\mathscr{G}} a|\leq\frac{k\rho}{n}(1-o(1))=o(1)$, where the inequality holds by substituting $p=\frac{k}{n}(1-o(1))$ and $|\partial_{\mathscr{G}} a|\leq\rho$, and the equality uses $\rho=o\big(\frac{n}{k}\big)$. We then use a Taylor expansion in \eqref{eq:mask_step_b}, followed by some simplifications.
    \item (d) uses the fact that $|\partial_{\mathscr{G}} a|\geq\frac{\rho}{\log n} = \omega(1)$ for each test $a$ in $\mathscr{G}_0$. 
    \item (e) uses the fact that $\frac{p\rho}{\log n}(1+o(1))=O\big(\frac{k\rho}{n\log n}\big)=\frac{n^{-\Omega(1)}}{\log n}<1$ (implying that a higher power in the exponent makes the overall expression smaller), and $|\partial_{\mathscr{G}_0} u| \le |\partial_{\mathscr{G}} u| \leq d^-$ by construction in our extraction procedure.
\end{itemize}

We now turn to the number of masked defective items in $\mathcal{B}$. Recall that for any two items $u,u'\in\mathcal{B}$, the events of being masked are independent due to the pairwise distances being at least 6. Furthermore, each item is defective independently with probability $p$ under $\mathbf{u}^*$, and the property of any given item being masked is independent of the item's defectivity status (as noted in \cite{Coj20}).  Hence, $|V_{1+}(\mathscr{G},\mathbf{u}^*)|$ stochastically dominates $\text{Binomial}\big(N,p\cdot\big(\frac{p\rho}{\log n}(1+o(1))\big)^{d^-}\big)$, which has an expectation of
\begin{align}
    &N\cdot p\cdot\bigg(\frac{p\rho}{\log n}(1+o(1))\bigg)^{d^-} \nonumber \\
    &=\frac{n^{1-2\beta(1-\theta)}}{\log^3 n}\cdot\frac{k}{n}\cdot\Big(\frac{k\rho}{n\log n}\Big)^{d^-}(1+o(1)) \\
    &=\frac{n^{1-(1-\theta)(2\beta+1)-d^-(1-\theta)(1-\beta)}}{(\log n)^{3+d^-}}(1+o(1)). \label{eq:V_1+_exp_bound0} \\
    &=\frac{n^{\Omega(1)}}{(\log n)^{3+d^-}}(1+o(1)), \label{eq:V_1+_exp_bound}
\end{align}
where \eqref{eq:V_1+_exp_bound0} substitutes $k = n^{\theta}$, $\rho = \big(\frac{n}{k}\big)^{\beta}$, and $p = \frac{k}{n}(1+o(1))$, and \eqref{eq:V_1+_exp_bound} holds since $d^-<\frac{1-(1-\theta)(2\beta+1)}{(1-\theta)(1-\beta)}$ by definition. By Lemma \ref{lem:mul_chernoff_bound} with $\epsilon=1-\frac{2(\log n)^{4+d^-}}{n^{\Omega(1)}}(1-o(1))$, we find that $|V_{1+}(\mathscr{G},\mathbf{u}^*)|>2\log n$ with probability $1-o(1)$. The same analysis can be repeated to show that $|V_{0+}(\mathscr{G},\mathbf{u}^*)|>2\log n$ with probability $1-o(1)$ (intuitively, the same readily follows for non-defectives because $k = o(n)$, i.e., the number of non-defectives is much higher).

Finally, we transfer our findings from the i.i.d.~prior $\mathbf{u}^*$ to the combinatorial prior $\mathbf{u}$.  Specifically, using Lemma \ref{lem:comb_vs_iid_prior}, we immediately obtain $|V_{1+}(\mathscr{G},\mathbf{u})|,|V_{0+}(\mathscr{G},\mathbf{u})|>\log n=\omega(1)$ with probability $1-o(1)$. By conditioning on $|V_{1+}(\mathscr{G},\mathbf{u})|=\omega(1)$ and $|V_{0+}(\mathscr{G},\mathbf{u})|=\omega(1)$, and applying Lemma \ref{lem:Pe_from_V+}, we get a conditional error probability of $1-o(1)$, which completes the first part of the proof of Theorem \ref{thm:rho_converse}.

\subsection{Part II} \label{sec:c=2}

We again work in the regime where $k = n^{\theta}$ and $\rho=\big(\frac{n}{k}\big)^\beta=n^{\beta(1-\theta)}$, with $\beta\in(0,1)$, i.e., the constant factors are set to one. We start by stating an auxiliary result that will be used later.

\begin{lemma} \label{lem:hypergeom_chernoff_bound} \textup{\cite[Section 5]{Ska13}}
For a random variable $Z_1\sim\textup{Hypergeometric}(n,z,k)$ (i.e., $\mathbb{P}[Z_1=z_1]={z\choose z_1}{n-z\choose k-z_1}{n\choose k}^{-1}$), $p=\frac{z}{n}$, and $0<t<p$, we have
\begin{align}
    \mathbb{P}[Z_1\leq(p-t)k]
    &\leq e^{-kD(p-t\|t)}
    \leq e^{-2t^2k},
\end{align}
where $D(a\|b)=a\log\frac{a}{b}+(1-a)\log\frac{1-a}{1-b}$ is the binary KL divergence.
\end{lemma}

We wish to show that with $T=(1-\epsilon)\frac{2n}{\rho}$,\footnote{We may assume equality for the purpose of proving a converse, since reducing the number of tests only makes recovery even more difficult.} where $\epsilon$ is a positive constant, no algorithm successfully recovers $\mathbf{u}$ from $(\mathbf{y},\mathscr{G})$ with probability $\Omega(1)$.  We proceed with a proof by contradiction, assuming that the opposite is true, i.e., that with $T=(1-\epsilon)\frac{2n}{\rho}$, there exists an algorithm that successfully recovers $\mathbf{u}$ from $(\mathbf{y},\mathscr{G})$ with probability $\Omega(1)$. Based on this assumption, we will show that $\epsilon = o(1)$, which gives us a contradiction.

\begin{lemma} \label{lem:degree_one}
    If $T=(1-\epsilon)\frac{2n}{\rho}$ and the success probability is $\Omega(1)$, then there are at least $2\epsilon n(1-o(1))$ degree-one items.
\end{lemma}

\begin{proof}
Let $\alpha_0 n$ and $\alpha_1 n$ be the number of items with degree zero and one respectively. We start by establishing that $\alpha_0=o(1)$.  This is because if there were $\Omega(n)$ degree-zero items, then with high probability there would be $\omega(1)$ degree-zero defectives and $\omega(1)$ degree-zero non-defectives (e.g., using the variant of Hoeffding's bound for the hypergeometric distribution \cite{Hoe63}).  Since these are impossible to distinguish from each other, this implies $o(1)$ success probability, in contradiction with the lemma assumption.

Now, by lower bounding the number of edges contributed by the items, and upper bounding the number of edges contributed by the tests, we obtain
\begin{align}
    &\alpha_1 n+2(1-\alpha_0-\alpha_1)n
    \leq\sum_{u\in V(G)}|\partial_{\mathscr{G}} u|
    =\sum_{a\in F(G)}|\partial_{\mathscr{G}} a| \nonumber \\
    &\qquad \leq T\rho
    =(1-\epsilon)2n.
\end{align}
Combining the left-most and right-most expressions and making $\alpha_1$ the subject and applying $\alpha_0 = o(1)$, we obtain $\alpha_1 \geq2\epsilon-2\alpha_0=2\epsilon(1-o(1))$.
\end{proof}

Before proceeding, we introduce some further notation:
\begin{itemize}
    \item $\mathcal{T}'$ denotes the set of tests with at least $\frac{\rho}{\log n}$ items of degree one, and $T'=|\mathcal{T}'|$.  We observe that the lower bound $T'\geq\frac{n}{\rho\log n}$ must hold, since otherwise the total number of items of degree one would be at most $(1-\epsilon) \frac{2n}{\rho} \cdot \frac{\rho}{\log n} + \frac{n}{\rho\log n} \cdot \rho \le \frac{3n}{\log n}$, contradicting Lemma \ref{lem:degree_one}.
    \item $T'_+$ denotes the number of tests in $\mathcal{T}'$ with exactly one defective item of degree one.
    \item $z$ denotes the number of items with degree one that appear in the tests of $\mathcal{T}'$. Note that $z$ satisfies $z\geq\frac{T'\rho}{\log n}\geq\frac{n}{\log^2n}$ because each test in $\mathcal{T}'$ has at least $\frac{\rho}{\log n}$ items of degree one, and $z\leq T'\rho$ because of the $\rho$-sized test constraint. These inequalities will be used in our analysis later.
    \item $Z_1$ denotes the number of {\em defective} items with degree one that appear in the tests of $\mathcal{T}'$. 
\end{itemize}
Note that among the four quantities introduced above, only $T'_+$ and $Z_1$ are random (i.e., affected by the randomness of the defective set). 

\begin{lemma} \label{lem:test_w_rho/logn_deg(1)_items}
When the number $T'$ of tests containing at least $\frac{\rho}{\log n}$ items of degree one is at least $\frac{n}{\rho\log n}$, any algorithm recovering $\mathbf{u}$ from $(\mathbf{y},\mathscr{G})$ has a success probability of $o(1)$.
\end{lemma}

Before proving this lemma, we state some additional auxiliary results.

\begin{lemma} \label{lem:V+_large}
Under the combinatorial prior with our assumed scaling laws, the size of the set $V_{0+}$ of masked non-defective items scales as $\omega(1)$ with probability $1-o(1)$.
\end{lemma}
\begin{proof}
We momentarily return to the i.i.d.~prior $\mathbf{u}^*$, where $p=\frac{k-\sqrt{k}\log n}{n}=\frac{k}{n}(1-o(1))$, and study how $T'_+$ scales. Under $\mathbf{u}^*$, the number of defective items of degree one in a test containing $\rho'\geq\frac{\rho}{\log n}$ items of degree one is distributed as $\text{Binomial}\big(\rho',\frac{k}{n}(1-o(1))\big)$, which implies
\begin{align}
    &\mathbb{P}[\text{$1$ def.~item of degree one}] \nonumber \\
    &={\rho'\choose 1}\Big(\frac{k}{n}(1-o(1))\Big)\Big(1-\frac{k}{n}(1-o(1))\Big)^{\rho'-1} \\
    &\stackrel{(a)}{\geq}\frac{\rho}{\log n}\Big(\frac{k}{n}(1-o(1))\Big)\Big(1-\frac{k}{n}(1-o(1))\Big)^{\rho-1} \\
    &\stackrel{(b)}{=}\frac{k\rho}{n\log n}\Big(1-\frac{k\rho}{n}(1-o(1))\Big)(1-o(1)) \\
    &\stackrel{(c)}{=}\frac{k\rho}{n\log n}(1-o(1)),
\end{align}
where (a) uses $\rho' \ge \frac{\rho}{\log n}$ and $\rho' \le \rho$, (b) uses $\frac{k\rho}{n}=o(1)$ and a Taylor expansion (followed by some simplifications), and (c) uses $\frac{k\rho}{n}=o(1)$. Observe that the event of a test in $\mathcal{T}'$ having at exactly one defective item of degree one is independent of the same event for the other tests in $\mathcal{T}'$, since the items are defective independently and cannot be shared across the tests (due to the degree of one). Hence, $T'_+$ stochastically dominates $\text{Binomial}\big(T',\frac{k\rho}{n\log n}(1-o(1))\big)$, which implies
\begin{align}
    \E\big[T'_+\big]
    &\geq T' \cdot \frac{k\rho}{n\log n} \cdot (1-o(1))
    =\frac{k}{\log^2n}(1-o(1)),
\end{align}
by applying $T'\geq\frac{n}{\rho\log n}$. By Lemma \ref{lem:mul_chernoff_bound} with $\epsilon=1-\frac{1+o(1)}{\log n}$, it follows that
\begin{align}
    \mathbb{P}\Big[T'_+\leq\frac{k}{\log^3n}\Big]
    &\leq\exp\Big(-\frac{k}{2\log^2n}(1+o(1))\Big)
    =o(1),
\end{align}
which implies that we have $T'_+>\frac{k}{\log^3n}=\omega(1)$ with probability $1-o(1)$.  This further implies that $|V_{0+}|>\frac{k}{\log^3n}\big(\frac{\rho}{\log n}-1\big)=\omega(1)$ (by counting the masked non-defective items with degree one in the $T'_+$ tests) with probability $1-o(1)$ under the i.i.d. prior. By Lemma \ref{lem:comb_vs_iid_prior}, it then follows that $|V_{0+}|=\omega(1)$ with probability $1-o(1)$ under the combinatorial prior.
\end{proof}

\begin{lemma} \label{lem:Z1_large}
Under the combinatorial prior, the number $Z_1$ of defective items with degree one that appear in the tests in $\mathcal{T}'$ scales as $\omega(1)$ with probability $1-o(1)$.
\end{lemma}

\begin{proof}
Due to fact that the defective set $\mathcal{K}$ is uniformly distributed, we have $\mathbb{P}[Z_1=z_1]={z\choose z_1}{n-z\choose k-z_1}{n\choose k}^{-1}$, i.e., $Z_1\sim\text{Hypergeometric}(n,z,k)$. Using this fact, we have
\begin{align}
    \mathbb{P}\bigg[Z_1\leq\frac{\E[Z_1]}{\log n}\bigg]
    &=\mathbb{P}\Big[Z_1\leq\frac{zk}{n\log n}\Big] \\
    &\stackrel{(a)}{\leq}\exp\bigg(-2\Big(\frac{z}{n}\Big(1-\frac{1}{\log n}\Big)\Big)^2k\bigg) \\
    &\stackrel{(b)}{\leq}\exp\bigg(-2\Big(\frac{1}{\log^2n}\Big(1-\frac{1}{\log n}\Big)\Big)^2k\bigg) \\
    &= \exp\bigg(-\frac{2k}{\log^4n}(1-o(1))\bigg), \label{eq:asdf}
\end{align}
where (a) uses Lemma \ref{lem:hypergeom_chernoff_bound} with $t=\frac{z}{n}\big(1-\frac{1}{\log n}\big)$, and (b) uses $z\geq\frac{n}{\log^2n}$. This proves that $Z_1>\frac{\E[Z_1]}{\log n}$ with probability $1-o(1)$. Since $\frac{\E[Z_1]}{\log n}=\frac{zk}{n\log n}\geq\frac{k}{\log^3n}$ (by applying $z\geq\frac{n}{\log^2n}$), we have $Z_1>\frac{k}{\log^4n}=\omega(1)$ with probability $1-o(1)$.
\end{proof}

With the auxiliary results in place, we turn to the proof of Lemma \ref{lem:test_w_rho/logn_deg(1)_items}.

\begin{proof}[Proof of Lemma \ref{lem:test_w_rho/logn_deg(1)_items}] Our aim is to show that with $T'\geq\frac{n}{\rho\log n}$, the success probability is $o(1)$. From this point onwards, we condition on $|V_{0+}|=\omega(1)$ and $Z_1=\omega(1)$, both occurring with probability $1-o(1)$ (see Lemma \ref{lem:V+_large} and Lemma \ref{lem:Z1_large}). Under these conditioned events, we have the following possible cases: (1) $|V_{1+}|=0$ and $|V_{0+}|=\omega(1)$; (2) $|V_{1+}|>0$ and $|V_{0+}|=\omega(1)$.
We proceed to analyze each case separately:
\begin{itemize}
    \item \textbf{Case 1:} 
    By assumption, we have $Z_1 = \omega(1)$, i.e., an unbounded number of defective items with degree one in the tests of $\mathcal{T}'$.  The condition $|V_{1+}|=0$ implies that the tests those $Z_1$ items are in cannot include any other defective items, meaning that there are $Z_1$ corresponding tests with exactly one degree-one defective, so that $T'_+ \ge Z_1 = \omega(1)$.  For each of these tests, the decoder can do no better than to guess the defective item uniformly at random from all degree-one items in the test, of which there are at least $\frac{\rho}{\log n}$. Hence, the success probability in each test is at most $\frac{\log n}{\rho} = o(1)$.
    \item \textbf{Case 2:} In this case, a direct application of Lemma \ref{lem:Pe_from_V+} gives us an error probability of $1-o(1)$ (i.e., a success probability of $o(1)$).
\end{itemize}
Since the success probability is $o(1)$ in both cases regardless of the specific values of $|V_+|=\omega(1)$ and $Z_1=\omega(1)$ being conditioned on, it follows that the unconditional success probability is also $o(1)$. This proves Lemma \ref{lem:test_w_rho/logn_deg(1)_items}.
\end{proof}


In view of Lemma \ref{lem:test_w_rho/logn_deg(1)_items}, for the assumed condition stated following Lemma \ref{lem:hypergeom_chernoff_bound} to hold (with an algorithm attaining success probability $\Omega(1)$), it must be the case that $T'<\frac{n}{\rho\log n}$.  We proceed by assuming that this is true, and showing that we arrive at a contradiction.

The condition $T'<\frac{n}{\rho\log n}$ implies that the number of edges contributed by the $T'$ tests is below $\frac{n}{\rho\log n}\cdot\rho=\frac{n}{\log n}=o(n)$, which further implies that $o(n)$ items of degree one participate in the tests of $\mathcal{T}'$. Recalling Lemma \ref{lem:degree_one}, this leaves us with greater than $2\epsilon n(1-o(1))$ items of degree one that participate in tests containing fewer than $\frac{\rho}{\log n}$ items of degree one. Comparing the overall degrees, we find that the number of such tests is greater than $\frac{2\epsilon n(1-o(1))}{\rho/\log n}$, and combining this with the assumption $T=(1-\epsilon)\frac{2n}{\rho}$ gives
\begin{align}
    &\frac{2\epsilon n(1-o(1))}{\rho/\log n}\leq(1-\epsilon)\frac{2n}{\rho} \nonumber \\
    &\qquad \implies\epsilon\leq\frac{1}{(\log n)(1-o(1))+1}=o(1), \label{eq:contradiction_ineq}
\end{align}
which gives the desired contradiction (we assumed $\epsilon = \Theta(1)$). This completes the second part of the proof of Theorem \ref{thm:rho_converse}.

\section{Conclusion}

In this paper, we have analyzed the performance of doubly-regular group testing designs in several settings of interest, with our main results summarized as follows:
\begin{itemize}
    \item In the unconstrained setting with sub-linear sparsity, the block-structured doubly-regular design with the DD algorithm matches the achievability result of the DD algorithm with near-constant tests-per-item, which is known to be optimal for $\theta \ge \frac{1}{2}$.
    \item In the unconstrained setting with linear sparsity, we complemented hardness results for exact recovery by deriving achievability results with only false negatives in the reconstruction.
    \item In the setting of $\rho$-sized tests with sub-linear sparsity and exact recovery, we improved on previously-known upper and lower bounds in regimes of the form $\rho = \Theta \big( \big( \frac{n}{k} \big)^{\beta} \big)$ with $\beta \in (0,1)$, complementing recent improvements that only hold for $\beta = 0$.
\end{itemize}
An immediate direction for future research is to establish converse bounds for approximate recovery in the linear regime with no false positives, and also to establish a more detailed understanding of the entire FPR vs.~FNR trade-off.  In addition, in the size-constrained setting, the optimal thresholds still remain unknown in general, despite the gap now being narrower.

\appendices

\section{Details of Conditional Independence Argument} \label{app:condition}

Here we provide a more mathematical description of the conditioning argument used in Step 3 in Section \ref{sec:ach_ana}.  Before defining the conditioning events for sub-matrices $\textsf{X}_1,\dotsc,\textsf{X}_r$, we first consider fixing several quantities:
\begin{itemize}
    \item The defective set $\mathcal{K}$ (e.g., $\mathcal{K}= \{1,\dotsc,k\}$);
    \item The masked non-defective set $\mathcal{G}$ (e.g., $\mathcal{G} = k+1,\dotsc,k+G$, where $G$ is the number of masked non-defectives);
    \item The set of positive tests $\mathcal{P}_1,\dotsc,\mathcal{P}_r$ for the $r$ sub-matrices (e.g., with $\textsf{X}_1$ corresponding to tests $1,\dotsc,\frac{n}{s}$, we could have $\mathcal{P}_1$ being a fixed size-$\frac{n}{2s}$ subset of those tests);
    \item The counts of masked defectives $M^1,\dotsc,M^r$ (e.g., $M^1 = 0$, $M^2 = 3$, etc.);
    \item The placements of non-masked non-defectives into negative tests (e.g., item $k+G+1$ is in the first negative test for $\textsf{X}_1$, item $k+G+2$ is in the fifth negative test for $\textsf{X}_7$, etc.).  We represent this by a set $\mathsf{FixedNeg}$ containing triplets $(i,j,t)$ with $i \in \{1,\dotsc,n\} \setminus (\mathcal{K} \cup \mathcal{G})$ indexing the (non-defective) item, $j \in \{1,\dotsc,r\}$ indexing the sub-matrix, and $t \in \{1,\dotsc,T\}$ indexing the (negative) test.
\end{itemize}
With the sets of positive tests $\{\mathcal{P}_j\}_{j=1}^r$ fixed, the corresponding sets of negative tests are also fixed, and are denoted by $\{\mathcal{N}_j\}_{j=1}^r$. 
Then, with $\mathcal{K}$, $\mathcal{G}$, $\{\mathcal{P}_j\}_{j=1}^r$, $\{\mathcal{N}_j\}_{j=1}^r$, $\{M^j\}_{j=1}^r$ and $\mathsf{FixedNeg}$ fixed, we can now list the conditioning events for $j=1,\dotsc,r$:
\begin{itemize}
    \item Let $\mathcal{A}_j$ be the event that each test in $\mathcal{P}_j$ contains at least one item from $\mathcal{K}$, and that each test in $\mathcal{N}_j$ contains no items from $\mathcal{K}$.
    \item Let $\mathcal{B}_j$ be the event that each test in $\mathcal{N}_j$ contains no items from $\mathcal{G}$.
    \item Let $\mathcal{C}_j$ be the event that there are exactly $M^j$ items in $\mathcal{K}$ whose (only) test in $\textsf{X}_j$ contains at least one other defective (i.e., it contains two or more in total);
    \item Let $\mathcal{D}_j$ be the event that for all $(i,t)$ satisfying $(i,j,t) \in \mathsf{FixedNeg}$, it holds that entry $(i,t)$ in $\textsf{X}_j$ equals one.
\end{itemize}
The final conditioning event is $\bigcap_{j=1}^r \big( \mathcal{A}_j \cap \mathcal{B}_j \cap \mathcal{C}_j \cap \mathcal{D}_j \big)$.  Then, since $ \mathcal{A}_j \cap \mathcal{B}_j \cap \mathcal{C}_j \cap \mathcal{D}_j$ is an event depending only on $\mathsf{X}_j$ (for $j=1,\dotsc,r$), we observe that the independence of the matrices $\{\mathsf{X}_j\}_{j=1}^r$ before conditioning is still maintained after conditioning, as desired.

\begin{figure*}[t]
    \normalsize
    \setcounter{mytempeqncnt}{\value{equation}}
    \begin{align}
        &\frac{(n-k)!(n-\frac{n\log2}{k})!}{(n-\frac{n\log2}{k}-k+1)!(n-1)!}-\frac{(n-k)!(n-\frac{n\log2}{k}-1)!}{(n-\frac{n\log2}{k}-k+1)!(n-2)!} \label{eq:cov_first_part_start} \\
        &\qquad=\frac{(n-k)!((n-\frac{n\log2}{k})!(n-2)!-(n-\frac{n\log2}{k}-1)!(n-1)!)}{(n-\frac{n\log2}{k}-k+1)!(n-1)!(n-2)!} \\
        &\qquad=\frac{(n-k)!(n-\frac{n\log2}{k}-1)!(n-2)!(n-\frac{n\log2}{k}-n+1)}{(n-\frac{n\log2}{k}-k+1)!(n-1)!(n-2)!} \\
        &\qquad\stackrel{(a)}{=}\frac{(n-\frac{n\log2}{k}-1)\dots(n-\frac{n\log2}{k}-k+2)}{(n-1)\dots(n-k+1)}\Big(-\frac{n\log2}{k}+1\Big) \\
        &\qquad=\Big(-\frac{n\log2}{k}(1+o(1)) \cdot \frac{1}{n-k+1}\Big)\prod_{i=1}^{k-2}\frac{n-\frac{n\log2}{k}-i}{n-i} \\
        &\qquad\stackrel{(b)}{=}-\frac{\log2}{k}(1+o(1))\Big(1-\frac{\log2}{k}(1+o(1))\Big)^{k-2} \\
        &\qquad\stackrel{(c)}{=} O\Big(\frac{1}{k}\Big), \label{eq:cov_first_part_end}
    \end{align}
    \hrulefill 
    \begin{align}
        &\frac{(n-k)!(n-k-\frac{n\log2}{k}+1)!(n-\frac{n\log2}{k}-1)!(n-\frac{2n\log2}{k})!}
        {(n-2)!(n-\frac{n\log2}{k}-1)!(n-k-\frac{n\log2}{k}+1)!(n-k-\frac{2n\log2}{k}+2)!}
        -\bigg(\frac{(n-k)!(n-\frac{n\log2}{k})!}{(n-1)!(n-k-\frac{n\log2}{k}+1)!}\bigg)^2 \label{eq:start_final} \\
        &\qquad=\frac{(n-k)!(n-\frac{2n\log2}{k})!}{(n-2)!(n-k-\frac{2n\log2}{k}+2)!}
        -\bigg(\frac{(n-k)!(n-\frac{n\log2}{k})!}{(n-1)!(n-k-\frac{n\log2}{k}+1)!}\bigg)^2 \\
        &\qquad\stackrel{(a)}{=}\frac{(n-\frac{2n\log2}{k})\dots(n-k-\frac{2n\log2}{k}+3)}{(n-2)\dots(n-k+1)}
        -\frac{(n-\frac{n\log2}{k})^2\dots(n-\frac{n\log2}{k}-k+2)^2}{(n-1)^2\dots(n-k+1)^2} \\
        &\qquad=\prod_{i=0}^{k-3}\frac{n-\frac{2n\log2}{k}-i}{n-i-2}-\prod_{i=0}^{k-2}\frac{(n-\frac{n\log2}{k}-i)^2}{(n-i-1)^2} \\
        &\qquad=\prod_{i=0}^{k-3}\Big(1-\frac{\frac{2n\log2}{k}-2}{n-i-2}\Big)-\prod_{i=0}^{k-2}\Big(1-\frac{\frac{n\log2}{k}-1}{n-i-1}\Big)^2 \\
        &\qquad=\Big(1-\frac{2\log2}{k}\Big(1+O\Big(\frac{k}{n}\Big)\Big)\Big)^{k-2}-\Big(1-\frac{\log2}{k}\Big(1+O\Big(\frac{k}{n}\Big)\Big)\Big)^{2(k-1)} \\
        &\qquad\stackrel{(b)}{=}\exp\bigg(-(2\log2)\Big(1+O\Big(\max\Big\{\frac{1}{k},\frac{k}{n}\Big\}\Big)\Big)\bigg) - \exp\bigg(-(2\log2)\Big(1+O\Big(\max\Big\{\frac{1}{k},\frac{k}{n}\Big\}\Big)\Big)\bigg) \\
        &\qquad\stackrel{(c)}{=}O\Big(\max\Big\{\frac{1}{k},\frac{k}{n}\Big\}\Big), \label{eq:end_final}
    \end{align}
    \hrulefill
    \vspace*{4pt}
\end{figure*}

\section{Proof of Lemma \ref{lem:covariance} (Covariance Calculation)} \label{app:covariance}

\noindent It suffices to show that the first and second parts of \eqref{eq:cov_M1j_M2j} simplify to $O\big(\max\{\frac{1}{k},\frac{k}{n}\}\big)$.

\textbf{First part:} Expanding the binomial coefficient in terms of factorials, the first term simplifies as in \eqref{eq:cov_first_part_start}--\eqref{eq:cov_first_part_end} at the top of the next page, where (a) expands all the factorials and simplifies, (b) uses $k =o(n)$, and (c) applies $\big(1-\frac{\log2}{k}(1+o(1))\big)^{k-2}=\exp(-(1+o(1))\log2)=\Theta(1)$.

\textbf{Second part:} Expanding the binomial coefficient in terms of factorials, we have \eqref{eq:start_final}--\eqref{eq:end_final} at the top of the next page, where (a) expands all the factorials and simplifies, (b) uses the fact that $1-a = e^{-a + O(a^2)} = e^{-a(1+O(a))}$ and $(1+O(a))(1+O(b)) = 1+O(\max\{a,b\})$ whenever $a,b=o(1)$, and (c) uses the fact that upon applying a Taylor expansion to each term, the leading terms cancel and only the first-order remainder remains.
    

{\bf Combining the terms:} Substituting $O\big(\frac{1}{k}\big)$ and $O\big(\max\big\{\frac{1}{k},\frac{k}{n}\big\}\big)$ into the first and second parts of \eqref{eq:cov_M1j_M2j} respectively, we obtain $\Cov\big[M_1^j,M_2^j\big] \le O\big(\max\big\{\frac{1}{k},\frac{k}{n}\big\}\big)$, as desired.

\section{False Positive Rate of COMP in the Linear Regime} \label{app:COMP}

Here we provide an analog of Theorem \ref{thm:linear_achievability} for the COMP algorithm (see Algorithm \ref{alg:COMP&DD_algo}), which comes essentially ``for free'' from our analysis of the DD algorithm.  However, we note that unlike our DD analysis, the result for COMP would also follow easily from prior work, particularly \cite{Ald20}.  Recall that the COMP algorithm only produces false positives, which implies that we only need to look at the FPR (i.e., FNR = 0). 

\begin{theorem} \label{thm:linear_achievability_COMP}
Using the COMP algorithm with the block-structured doubly-regular design with fixed parameters $s$ and $r$, when there are $k=pn$ defective items with constant $p\in(0,1)$, we have $\textup{FPR} \le \textup{FPR}_{\max}(1+o(1))$ with probability $1-o(1)$, where 
\begin{align}
    \textup{FPR}_{\max}=\big(1-(1-p)^{s-1}\big)^r. \label{eq:linear_achievability_COMP}
\end{align}
\end{theorem}
\begin{proof}
    From \eqref{eq:E[PD_i]_general}, we have
    \begin{align}
        \text{FPR} 
        &=\bigg(1-\prod_{i=1}^{s-1}\Big(1-\frac{k}{n-i}\Big)\bigg)^r \nonumber \\
        &=\bigg(1-\Big(1-\frac{k}{n(1-o(1))}\Big)^{s-1}\bigg)^r \nonumber \\
        &= \big(1-(1-p)^{s-1}\big)^r(1+o(1)),
    \end{align}
    where the second and third qualities use the fact that $s$, $r$, and $p$ are all constant with respect to $n$.
\end{proof} 

A given FPR value corresponds to an average of $(n-k)\text{FPR}$ false positives, which may potentially be much larger than $k$.  To place the number of false positives and the actual number of defectives on the ``same scale'', we find it more convenient to work with the normalized quantity $\text{FPR} \frac{n-k}{k} =  \text{FPR} \frac{1-p}{p}$.  Then, this quantity equaling a given value $\alpha > 0$ corresponds to an average of $\alpha k$ false positives. 

Recalling the notion of rate in Definition \ref{def:rate}, we have the following analog of Corollary \ref{cor:small_p}; the proof is similar but simpler, so is omitted.

\begin{corollary} \label{cor:small_p_COMP}
    Under the setup of Theorem \ref{thm:linear_achievability_COMP}, there exist choices of $r$ and $s$ (depending on $p$) such that, in the limit as $p \to 0$ (after having taking $n \to \infty$), we have (i) $\mathrm{FPR}_{\max} \frac{1-p}{p} \rightarrow 0$, (ii) the rate approaches $\log 2$, and (iii) it holds that  $s=\frac{\log2}{p}(1+o(1))$ and $r = \frac{\log(\frac{1}{p})}{\log2}(1+o(1))$.
\end{corollary}

Similarly to the discussion following Corollary \ref{cor:small_p}, this result is consistent with the rate of $\log 2$ attained for COMP with approximate recovery in the sub-linear regime $k = o(n)$ \cite[Sec.~5.1]{Ald19}. 

\begin{figure}[!t]
    \centering
    \includegraphics[width=0.475\textwidth]{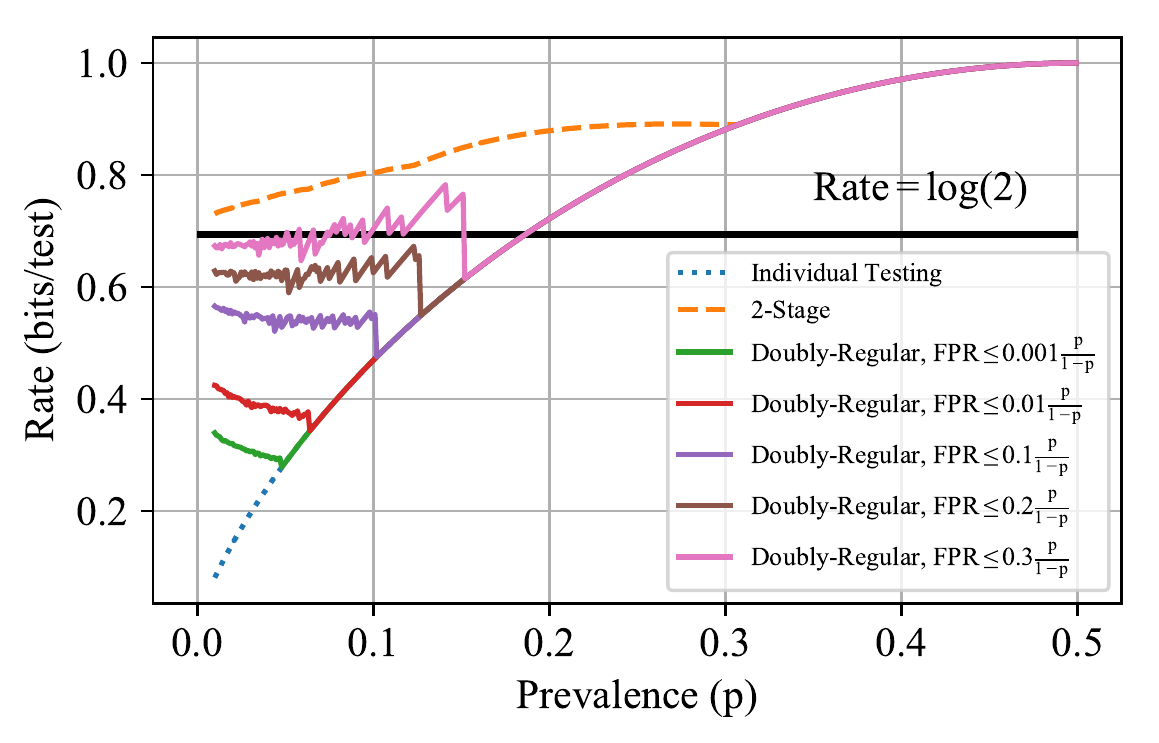}
    \caption{Achievable rates for COMP decoding with the doubly-regular design and approximate recovery (along with individual testing and a two-stage design \cite{Ald20}, both of which attain exact recovery).} \vspace*{-1.5ex}
    \label{fig:doub_reg_rate_COMP}
\end{figure} 

\begin{figure}[!t]
    \centering
    \includegraphics[width=0.45\textwidth]{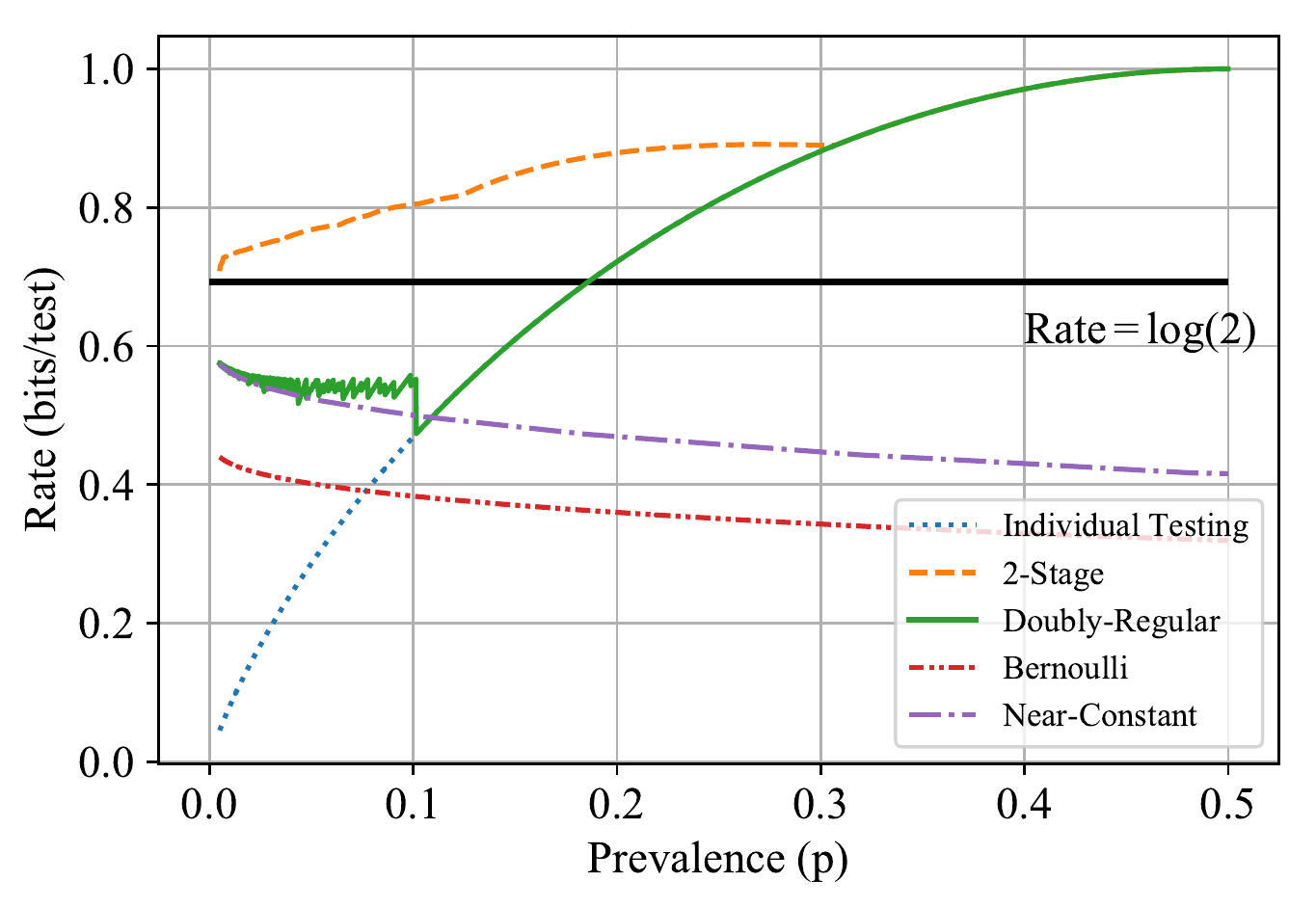} \quad 
        \includegraphics[width=0.45\textwidth]{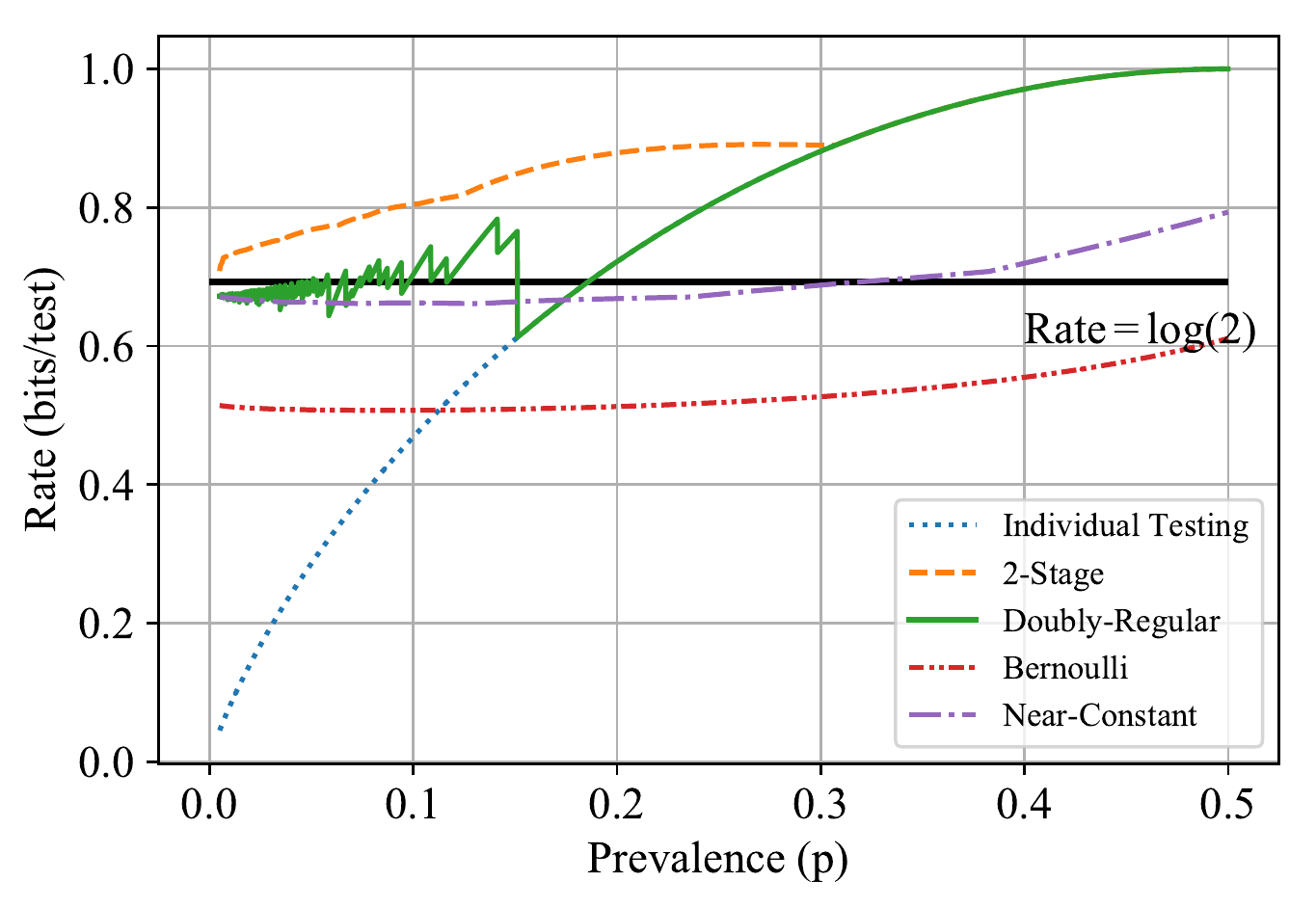} 
    \caption{Comparison of rates for COMP decoding with various tests designs, for approximate recovery with $\alpha = 0.1$ (Left) and $\alpha = 0.3$ (Right); the rates for Bernoulli designs and near-constant column weight designs are given in \cite{Bay20a}.  Individual testing and the two-stage design \cite{Ald20} attain exact recovery.} \vspace*{-1.5ex}
    \label{fig:comparisons}
\end{figure} 

In this case, we can easily argue that the constant of $\log 2$ is optimal.  To see this, note that if we can could attain at most $\alpha k$ false positives on average for arbitrarily small $\alpha > 0$, then we could use this to construct a two-stage adaptive group testing algorithm where the second stage tests the items outputted in the first stage individually, thus using an average of $(1+\alpha)k$ tests or fewer.  When $p$ approaches zero, these additional tests contribute a arbitrarily small fraction compared to the leading $\Theta\big( k \log \frac{n}{k} \big) = \Theta\big(k\log\frac{1}{p}\big)$ term, so the two-stage design attains zero error probability with an arbitrarily small increase in the rate.  However, the converse result of \cite{Ald20} shows that rates above $\log 2$ are impossible in this two-stage setting (see also \cite{Mez11} for the sublinear regime).  This establishes the optimality of the constant $\log 2$ above.

To visualize the constant-$p$ regime, we plot the rates attained by the COMP algorithm with the doubly-regular design in Figure \ref{fig:doub_reg_rate_COMP} (see the text following Corollary \ref{cor:small_p} for discussion on why the behavior of the curves reaching $\log 2$ as $p \to 0$ is not visible; a similar discussion applies here).  In Figure \ref{fig:comparisons}, we additionally compare against other random non-adaptive test designs, for which bounds on the FPR were given in \cite{Bay20a}.  We observe that the doubly-regular design consistently outperforms the Bernoulli design, and also outperforms the near-constant column weight design except near certain $p$ values where the doubly-regular curve is discontinuous.

\bibliographystyle{IEEEtran}
\bibliography{JS_References}

\begin{IEEEbiographynophoto}{Nelvin Tan}
    received the B.Comp.~degree in computer science and statistics from the National University of Singapore in 2021. He is currently pursuing the Ph.D.~degree from the Signal Processing and Communications Group in the Department of Engineering, University of Cambridge.  His research interests include information theory and high-dimensional statistics.
\end{IEEEbiographynophoto}

\begin{IEEEbiographynophoto}{Way Tan}
    received a double degree in mathematics and computer science from the National University of Singapore in 2021.  His research interests include information theory and applied mathematics.
\end{IEEEbiographynophoto}

\begin{IEEEbiographynophoto}{Jonathan Scarlett}
    (S'14 -- M'15) received 
    the B.Eng.~degree in electrical engineering and the B.Sci.~degree in 
    computer science from the University of Melbourne, Australia. 
    From October 2011 to August 2014, he
    was a Ph.D. student in the Signal Processing and Communications Group
    at the University of Cambridge, United Kingdom. From September 2014 to
    September 2017, he was post-doctoral researcher with the Laboratory for
    Information and Inference Systems at the \'Ecole Polytechnique F\'ed\'erale
    de Lausanne, Switzerland. Since January 2018, he has been an assistant
    professor in the Department of Computer Science and Department of Mathematics,
    National University of Singapore. His research interests are in
    the areas of information theory, machine learning, signal processing, and
    high-dimensional statistics. He received the Singapore National Research Foundation (NRF)
    fellowship, and the NUS Early Career Research Award.
\end{IEEEbiographynophoto}

\end{document}